\documentclass[onecolumn,draftclsnofoot,12pt]{IEEEtran}

\usepackage[nolist,nohyperlinks]{acronym}
\begin{acronym}
    \acro{AP}{access point}
    \acro{BPP}{binomial point process}
    \acro{BLP}{binomial line process}
    \acro{BLCP}{binomial line Cox process}
    \acro{CCDF}{complementary cumulative density function}
	\acro{CDF}{cumulative density function}
    \acro{PDF}{probability density function}
    \acro{PPP}{Poisson point process}
    \acro{BS}{base station}
    \acro{LOS}{line of sight}
    \acro{SINR}{signal-to-interference-plus-noise ratio}
    \acro{SNR}{signal to noise ratio}
    \acro{mm-wave}{millimeter wave}
    \acro{PV}{Poisson-Voronoi}
    \acro{LOS}{line-of-sight}
    \acro{NLOS}{non line-of-sight}
    \acro{PGFL}{probability generating functional}
    \acro{PLCP}{Poisson line Cox process}
    \acro{MBS}{macro base station}
    \acro{PLT}{Poisson line tessellation}
    \acro{SBS}{small cell base station}
    \acro{RAT}{radio access technique}
    \acro{5G}{fifth generation}
    \acro{PLP}{Poisson line process}
    \acro{UE}{user equipment}
\end{acronym}

\usepackage{subfig}
\usepackage{xcolor}
\usepackage{float}
\usepackage{svg}
\usepackage{pdfpages}
\usepackage{comment}

\usepackage{graphicx,wrapfig,lipsum}
\usepackage{rotating}
\usepackage{amsmath, amssymb, amsthm}
\usepackage{stmaryrd}
\usepackage{tikz}
\usepackage{verbatim}
\usepackage{url}
\usepackage[utf8]{inputenc}
\usepackage[english]{babel}
\newtheorem{theorem}{Theorem}[]
\newtheorem{corollary}{Corollary}[]

\newtheorem{remark}{Remark}[]

\newtheorem{lemma}[]{Lemma}

\usepackage{multicol}

\usepackage{algorithm}
\usepackage{algpseudocode}
\usepackage{scalerel}
\usepackage{multicol}
\usepackage{setspace}
\newtheorem{definition}{Definition}
\usepackage[noadjust]{cite}
\usepackage{booktabs}
\usepackage{bbm}
\usepackage[normalem]{ulem}
\usepackage{color}
\usepackage{dsfont}
\usepackage{bm}
\usepackage{setspace}
\usetikzlibrary{automata}
\usetikzlibrary{arrows}
\usetikzlibrary{shapes.geometric, arrows}
\usepackage[]{footmisc}
\usetikzlibrary{positioning}
\usetikzlibrary{calc}
\usetikzlibrary{arrows}
\allowdisplaybreaks

\newcommand\blfootnote[1]{%
  \begingroup
  \renewcommand\thefootnote{}\footnote{#1}%
  \addtocounter{footnote}{-1}%
  \endgroup
}

\usepackage[margin=0.85cm]{caption}
\usepackage{mathtools}
\usepackage{csquotes}
\usepackage{geometry}
\title{ Binomial Line Cox Processes: Statistical Characterization and Applications in Wireless Network Analysis} 

\author{Mohammad Taha Shah, {\it Graduate Student Member, IEEE}, Gourab Ghatak, {\it Member, IEEE}, Souradip Sanyal, and Martin Haenggi, {\it Fellow, IEEE}
\thanks{\footnotesize{Mohammad Taha Shah is with the Bharti School of Telecommunication Technology and Management, IIT Delhi (email: bsz218183@iitd.ac.in). Gourab Ghatak is with the Department of Electrical Engineering, IIT Delhi, India (email: gghatak@ee.iitd.ac.in). Souradip Sanyal is with Continental AG. He was with the Department of Electronics and Communication Engineering IIIT-Delhi at the time of this work (email: souradip20156@iiitd.ac.in). Martin Haenggi is with the Department of Electrical Engineering, University of Notre Dame, Notre Dame, IN 46556 USA (e-mail: mhaenggi@nd.edu).}}}

\begin{document}

\maketitle



\begin{abstract}
The current analysis of wireless networks whose transceivers are confined to streets is largely based on Poissonian models, such as Poisson line processes and Poisson line Cox processes. We demonstrate important scenarios where a model with a finite and deterministic number of streets, termed binomial line process, is more accurate. We characterize the statistical properties of the BLP and the corresponding binomial line Cox process and apply them to analyze the performance of a network whose access points are deployed along the streets of a city. Such a deployment scenario will be typical for 5G and future wireless networks. In order to obtain a fine-grained insight into the network performance, we derive the meta distribution of the signal-to-interference and noise ratio. Accordingly, we investigate the mean local delay in transmission and the density of successful transmission. These metrics, respectively, characterize the latency and coverage performance of the network and are key performance indicators of next-generation wireless systems.
\end{abstract}


\IEEEkeywords Stochastic geometry, Line processes, Binomial point processes, Cox process, Meta Distributions.

\section{Introduction}
\subsection{Motivation}
\blfootnote{The codes for generating the numerical results of this paper are available for download~\cite{codes}.}
Line processes are useful statistical tools for studying a variety of engineering problems such as transportation and urban infrastructure planning, wireless communications, and industrial automation scenarios~\cite{ripley1976foundations, baccelli1997stochastic}. In the two-dimensional Euclidean plane, a line process, is a random collection of lines whose locations and orientations reside in a parameter space (to be defined shortly) according to a spatial stochastic process. Leveraging a line process, researchers often study doubly-stochastic processes called Cox processes, which are Poisson point processes defined with the line process as their restricting domain~\cite{dhillon2020poisson}. These models are key in deriving insights into engineering and planning questions such as
\begin{itemize}
    \item {\it How many electric vehicle charging points does a city need to have along the streets?}
    \item {\it From a typical urban home, how far is the nearest bus stop? Accordingly, what should be the density of bus stops in a given city?}
    \item {\it For on-road deployments of wireless small cells as envisaged in 5G and future networks, what would be the cellular coverage performance of a user in the city?}
\end{itemize}
In the context of wireless communications in particular, AT\&T plans to deploy compact Ericsson street radio small cells integrated with the street lights to accelerate 5G deployments~\cite{press}. Along the same lines, New York City has already deployed such small cells focusing on high-speed connectivity for pedestrian users~\cite{press3}. Such deployments seamlessly integrate \acp{AP} in the urban infrastructure to provide ubiquitous connectivity to the users. In this regard, the network coverage of any point in the network for such on-street deployed small cells is a major point of interest for network operators.

\subsection{Related Work}
In the literature, several line processes have been used to study the structure of streets, e.g., the \ac{PLP}, the Poisson Voronoi tessellation (PVT), and the Poisson lilypond model ~\cite{kendall2017random, kahn2016improper, jeyaraj2020cox, chenavier2016extremes}. Out of all the candidate line processes, the most popular one for modeling streets, especially in wireless network studies, is the \ac{PLP}~\cite{chetlur2020stochastic, jeyaraj2020cox}. The statistical properties of the \ac{PLP} provide significant tractability in developing insights into the engineering problems of interest. For example, in \cite{small2019ghatak}, the authors have employed the \ac{PLP} and \ac{PLCP} for studying on-street deployment of small cell \acp{BS} and derived the performance of a pedestrian user in a multi-tier multi \ac{RAT} cellular network. Similarly, in \cite{chetlur2019coverage, chetlur2020load}, the authors have characterized the performance of vehicular networks based on the \ac{PLP} model for streets and consequently, \ac{PLCP} for emulating vehicular nodes. However, one major drawback of the \ac{PLP} is that it fails to accurately take into account some salient features of urban street networks, e.g., finite street lengths, T-junctions, and varying density of streets across a given city.\footnote{From an infrastructure planning perspective, the density of the streets is measured in terms of the total street length} per unit area. The recent work by Jeyaraj {\it et al.}~\cite{jeyaraj2020cox} proposed a generalized framework for Cox models to study vehicular networks. The authors significantly improved the accuracy of line process models to account for finite street lengths by considering T-junctions, stick processes, and Poisson lilypond models. However, their work does not consider varying line densities from the perspective of a single city. This is an important aspect of applications such as urban transport networks and wireless deployment planning, in which streets near the city center are denser as compared to the streets in the suburbs. Such characteristics of an urban scenario can only be studied using non-homogeneous line processes which is the focus of this paper.

To illustrate this, the street network of Paris is shown in Fig.~\ref{fig:Paris}, where, as we move out of the city center, the street density decreases. More precisely, in~\cite{orangelabs}, researchers have emulated the streets of Lyon with multi-density Poisson-Voronoi tessellations (PVTs) and Poisson line tessellations (PLTs). Details on the mathematical framework to develop these models can be found in~\cite{PhysRevE.83.036106}. As shown in Fig.~\ref{fig:PVT21}, the authors tuned the parameters of the PVT model to emulate the distinct street structures in different parts of the city, e.g., the red-colored zone corresponds to a \ac{PPP} of intensity 107 km$^{-2}$ as compared to the green (\ac{PPP} with intensity 50 km$^{-2}$) or purple. In contrast, in this paper, we develop a novel stochastic line process where for a fixed set of parameter values we can emulate distinct street structures in the city center and the suburbs.

\begin{figure}
    \centering
    \subfloat[]{\includegraphics[height = 0.3\textwidth, width = 0.4\textwidth]{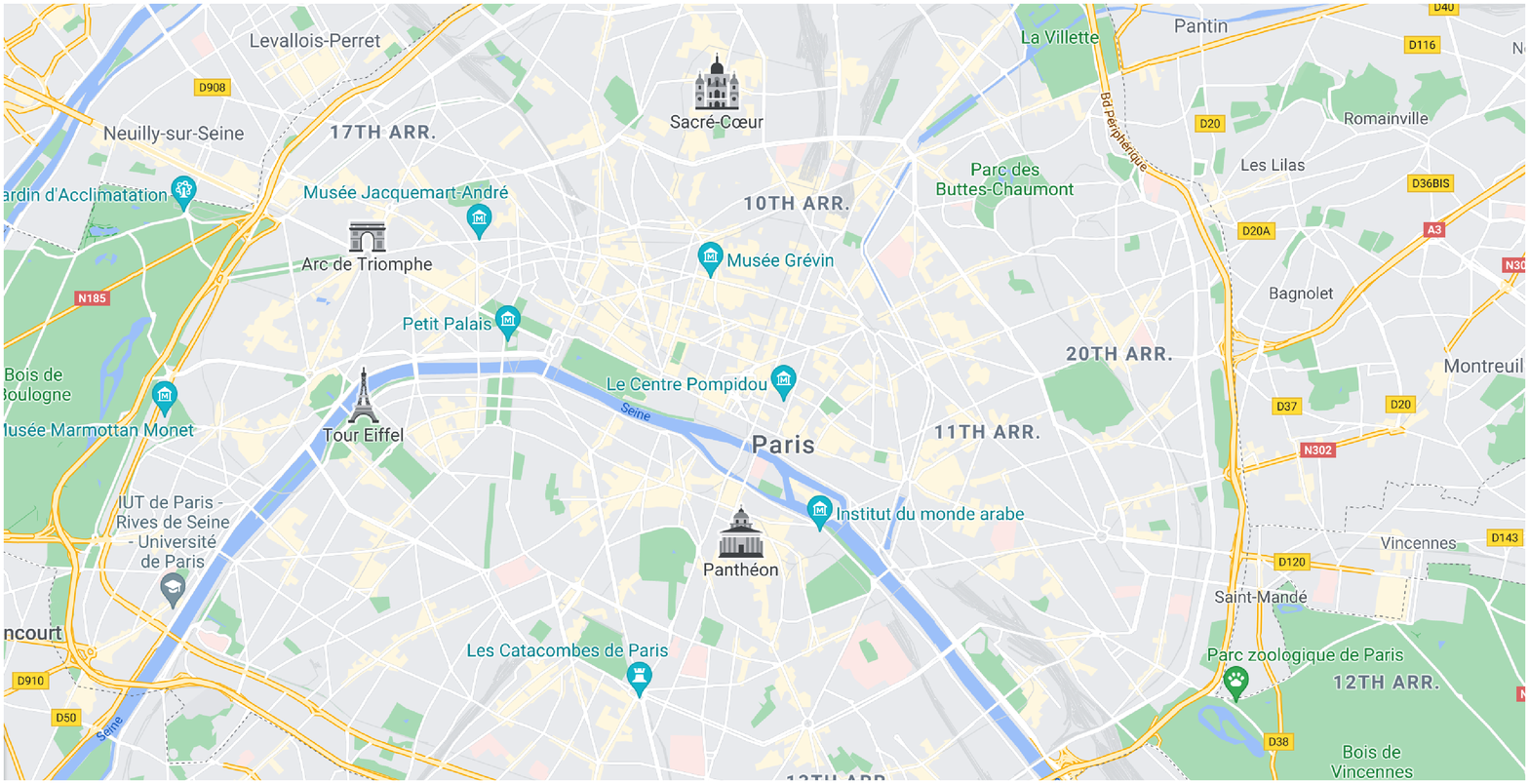}
    \label{fig:Paris}}
    \hfil
    \subfloat[]{\includegraphics[height = 0.3\textwidth, width = 0.4\textwidth]{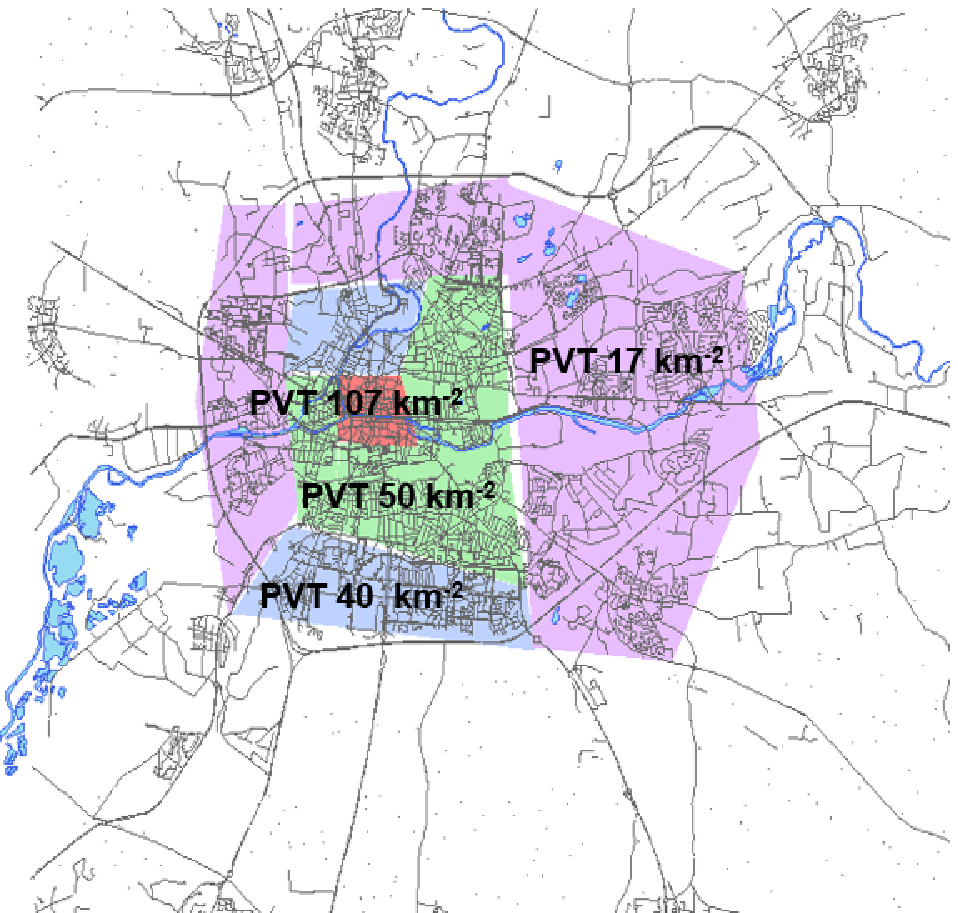}
    \label{fig:PVT21}}
    \caption{(a) Part of Paris city map. Visual inspection reveals the inhomogeneity of the street density (b) PVT models for fitting for streets of Lyon~\cite{orangelabs}. For a PVT, the model parameter is the density of the underlying \ac{PPP}.}
\end{figure}

Furthermore, since the city administrations sanction the construction of streets on a case-by-case basis, it is more likely that a deterministic number of streets characterizes the street network in a city. In other words, for initial infrastructure planning, one may be interested in the questions posed above {\it given} that there is a certain number, say $n_{\rm B}$, of streets in a city than for a given density of streets. This prompts us to study the properties of the street network parameterized by $n_{\rm B}$.

To address these aspects of urban street networks, recently in a short note, we have introduced the \ac{BLP} as a new stochastic line process that consists of a deterministic number of lines generated in a bounded generating set~\cite{ghatak22line}. Let us visualize the difference in the network formed by a \ac{BLP} as compared to a \ac{PLP}, with the help of realizations of these two candidate line processes (please refer to Section~\ref{sec:Cons} for details on the construction of a \ac{BLP}). In Fig.~\ref{fig:blp}, we plot a realization of a BLP where $n_{\rm B} = 10$ streets are generated within a distance of $R = 100$ from the origin. We compare it with a realization of a PLP in Fig.~\ref{fig:plp2} where the generating PPP has an intensity $n_{\rm B}/(2 \pi R) = 0.016$. It is apparent that the BLP captures the varying street and street length densities of a single city by virtue of its non-stationarity while the PLP-based models are stationary and thus restricted to a single fixed density. Consequently, the BLP model is necessary to contrast the performance of downtown users and sub-urban users.


\begin{figure}
\centering
\subfloat[]
{\includegraphics[width=0.4\textwidth]{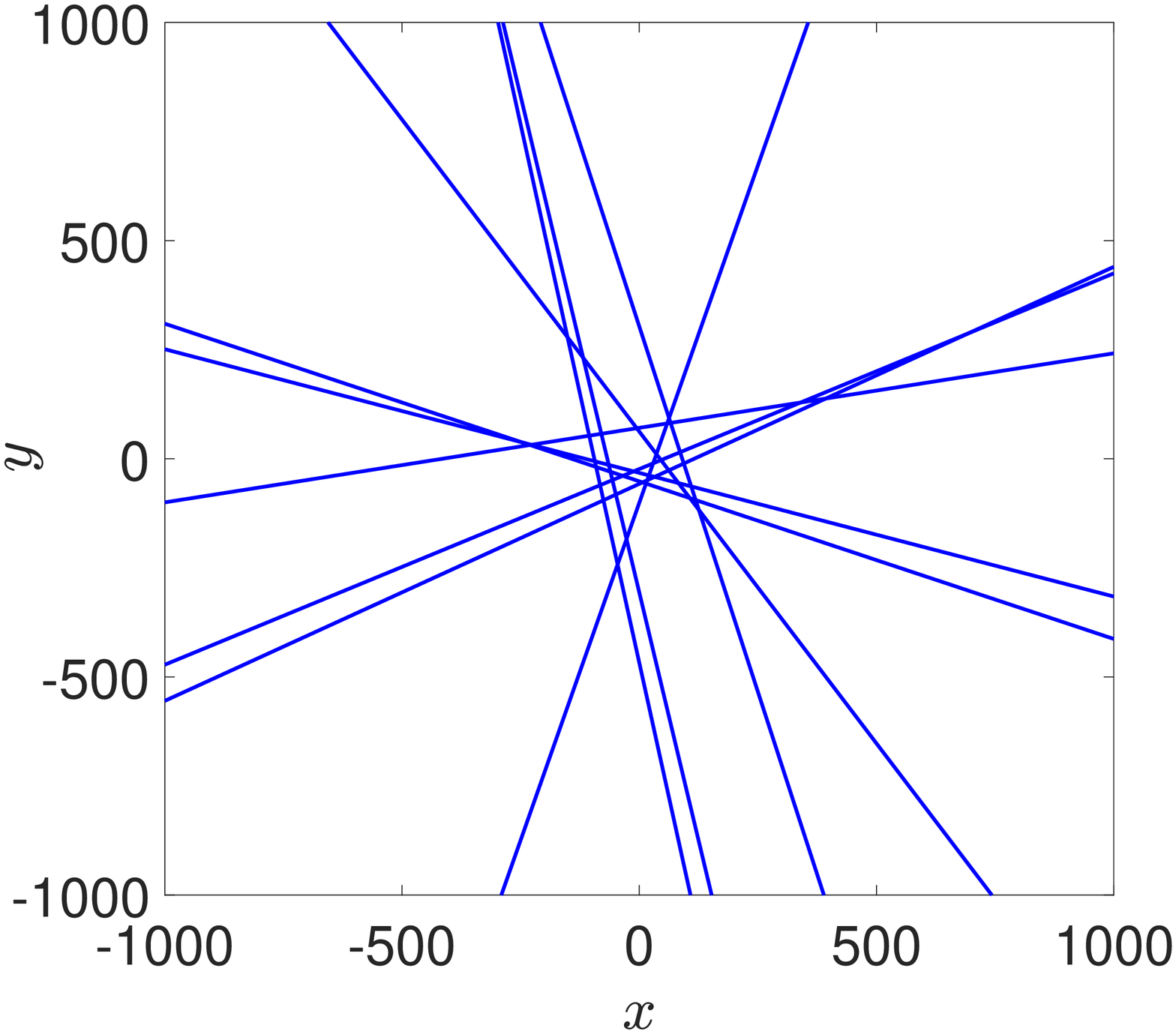}
\label{fig:blp}}
\hfil
\subfloat[]
{\includegraphics[width=0.4\textwidth]{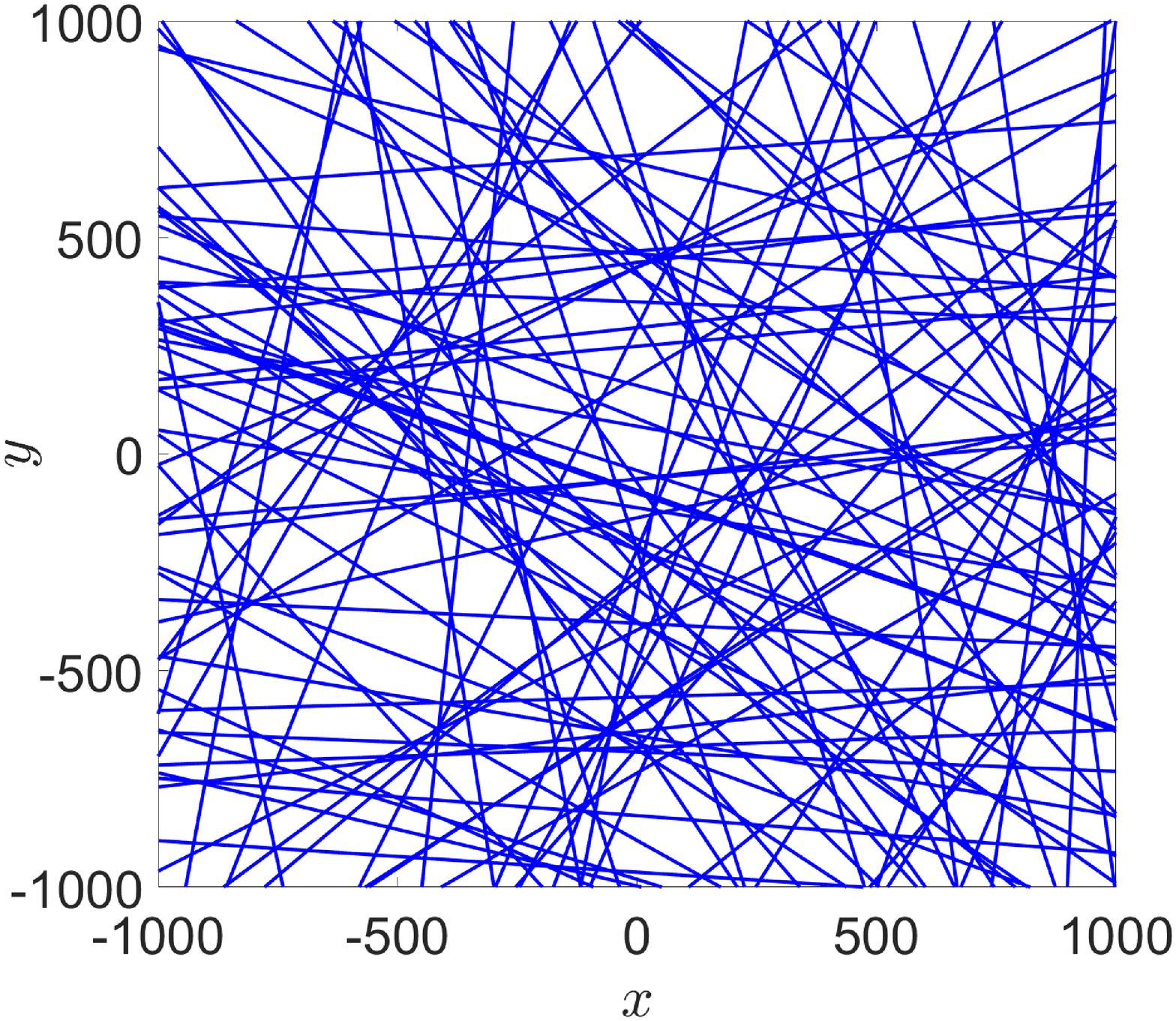}
\label{fig:plp2}}
\caption{(a) A BLP with $n_{\rm B} = 10$ and $R = 100$ (b) A PLP whose underlying PPP has intensity $\frac{n_{\rm B}}{2\pi R}$ with $n_{\rm B} = 10$ and $R = 100$. Note: Here $R$ is the radius of the circle in which BLP lines are generated.}
\label{fig:lines} 
\end{figure}

\subsection{Contributions and Organization}
The main contributions of this paper are summarized as follows:
\begin{itemize}
    \item We characterize the recently introduced \ac{BLP} and \ac{BLCP} models and give a thorough analysis of several radial density characteristics, including the line length radial density and intersection density of the BLP and PLP. We also derive the distance distribution to the nearest intersection from a test point on the \ac{BLP}.
    \item We derive the \ac{PGFL} of BLCP and provide a simpler derivation (as compared to~\cite{ghatak22line}) of the distribution to the nearest \ac{BLCP} point from a test point in the Euclidean plane.
    \item Leveraging the \ac{PGFL}, we characterize the transmission success probability in a wireless network where the {\acp{AP}} form a \ac{BLCP}. Such a model can emulate a variety of network scenarios, e.g., Wi-Fi \acp{AP} in an industrial automation setup, on-street deployment of 5G small cells, etc. In contrast, the commonly used stationary models such as the \ac{PLCP} are restricted to the homogeneous setting where the transceiver density is fixed everywhere.
    \item To obtain fine-grained insight into a wireless network modeled by a \ac{BLCP}, we characterize the meta distribution \cite{haenggi2021meta1}, \cite{haenggi2021meta2} of the \ac{SINR}. In particular, we first derive the moments of the conditional success probability and employ them to study the mean local delay and the density of successful transmission in the wireless network, which provides insight on the latency and coverage of the network. Second, the Gil-Pelaez inversion theorem results in the exact characterization of the meta distribution. 
\end{itemize}
The rest of the paper is organized as follows. The mathematical preliminaries regarding the construction and statistical characterization of the \ac{BLP} are presented in Section~\ref{sec:pre}, along with the derivation of radial and intersection density, and nearest intersection distance distribution. In Section~\ref{sec:BLCP} we define and characterize the \ac{BLCP}, its void probability, and the \ac{PGFL} of the \ac{BLCP}. Section~\ref{sec:app} presents an application of the derived framework in studying the transmission success probability in a wireless network. The meta distribution is derived in Section~\ref{sec:meta}. In Section~\ref{sec:NRD} we present numerical results to discuss the salient features of a wireless network based on a BLCP location model of the \acp{AP}. Finally, the paper concludes in Section~\ref{sec:con}.


\begin{table}[t]
\centering
\small{
\begin{tabular}{|l|l|}
\hline
Notation         & Description \\ \hline \hline
$\mathcal{L}$ & Binomial line process\\
$L_i$ & $i-$th line of a BLP\\
$R$ & Radius of circle in which BLP $\mathcal{L}$ is generated \\
$\mathcal{B}((\theta,r),t)$ & A disk of radius $t$ centered at $(\theta,r)$ \\
$n_{\rm B}$ & Number of lines of the BLP $\mathcal{L}$ \\
$\mathcal{D}_{\rm B}(r_0, t)$ & Domain band corresponding to $\mathcal{B}((0,r_0), t)$\\
$A_{\rm D} (r_0,t)$ & Area of the domain bands corresponding to $\mathcal{B}((0,r_0), t)$\\
$\mathcal{V}_{\rm BLP}(n_{\rm B}, \mathcal{B})$& Void probability of the BLP with $n_{\rm B}$ lines on $\mathcal{B}$\\
$\mathcal{V}_{\rm BLCP}(n_{\rm B}, \mathcal{B})$& Void probability of the BLCP with $n_{\rm B}$ lines on $\mathcal{B}$\\
$\rho_S$ & Density of length of chords/line segments in a bounded Borel set $S$ \\
$\rho_i(w)$ & Density of length of chords/line segments in the $i-$th annulus of equal width $w$ \\
$\mathcal{R}$ & Line length measure \\
$\rho(r)$ & Line length density \\
$\mathcal{N}$ & Average number of intersections within a disk centered at the origin \\
$\rho_{{\times}, i}(w)$ & Density of intersections in the $i-$th annulus of equal width $w$ \\
$\mathcal{R_{\times}}$ & Intersection measure of BLP\\
$\rho_{\times}(r)$ & Intersection density of BLP\\
$\rho_{\rm p}(\lambda_{\rm PPP})$ & Intersection density of PLP\\
$A_{D_{\bf I}} (r_0,t)$ & Area of the domain bands corr. to the nearest intersection from the test point\\
$\Phi_i$ & 1D PPP on $L_i$\\
$\lambda$ & Density of $\Phi_i$\\
$C(\theta, r)$ & Length of a chord generated by a line corresponding to $(\theta, r)$ on $\mathcal{B}((0,r_0), t)$.\\
$\Phi$ & Binomial line Cox process with parameters $n_{\rm B}$ and $R$\\
$\xi(r_0)$ & SINR received at the test point \\
$p_{\rm S}(\gamma)$ & Success probability at threshold $\gamma$\\
$\mathcal{P}_{\rm M}(\gamma,\beta)$ & Meta distribution evaluated at SINR threshold $\gamma$ and reliability threshold $\beta$ \\ \hline
\end{tabular}
}
\caption{Summary of notations used in the paper.}
\label{tab:notations}
\end{table}

\subsection{Notations}
Line processes are denoted by calligraphic letters, e.g., $\mathcal{L}$, while point processes are denoted by $\Phi$. An instance of a point process is denoted by $\phi$. A BLP is presumed to consist of $n_{\rm B}$ lines unless stated otherwise. Capital math-font characters are random variables unless otherwise stated, e.g., $Z$, while an instance of $Z$ is denoted by $z$. Similarly, other scalar quantities are denoted by small letters, e.g., $d$. In the two-dimensional Euclidean plane, $\mathbb{R}^2$, locations are denoted by bold font, e.g., ${\bf x}$. $\mathbb{E}$ and $\mathbb{P}$ represent the expectation and the probability operators, respectively. The other notations used in the paper are summarized in Table~\ref{tab:notations}.


\section{Preliminaries}
\label{sec:pre}
We will restrict our discussion to $\mathbb{R}^2$ due to its relevance to the locations of users and \acp{AP} in wireless networks.
\subsection{Binomial Line Process - Construction}
\label{sec:Cons}
A \ac{BLP} $\mathcal{L}$ is a finite collection of lines in the two-dimensional Euclidean plane. Formally, 
\begin{align}
\mathcal{L} \subset Q \triangleq \bigcup_{r \in \mathbb{R} \cap [-R,R], \theta \in [0,\pi)} \{(x,y) \in \mathbb{R}^2 \colon x\cos \theta + y\sin \theta=r \}.    
\end{align}
Each line of $\mathcal{L}$ corresponds to a point of a \ac{BPP} defined on the finite cylinder $\mathcal{D}:=$ [$0,\pi$) $\times$ $[-R, R]$. We will call $\mathcal{D}$ the generating set or the domain set of $\mathcal{L}$, and a point $(\theta_i, r_i)  \in \mathcal{D}$, corresponding to a line $L_i \in \mathcal{L}$, the generating point of $L_i$. The line segment drawn from the origin to $(\theta_i, r_i)$ forms the normal to the line $L_i$. The pair of coordinates $(\theta_i, r_i)$ of any point in $\mathcal{D}$ has a bijective mapping to a line $L_i$ in $\mathbb{R}^2$. 
\begin{figure}
    \centering
    \includegraphics[width = 0.5\textwidth]{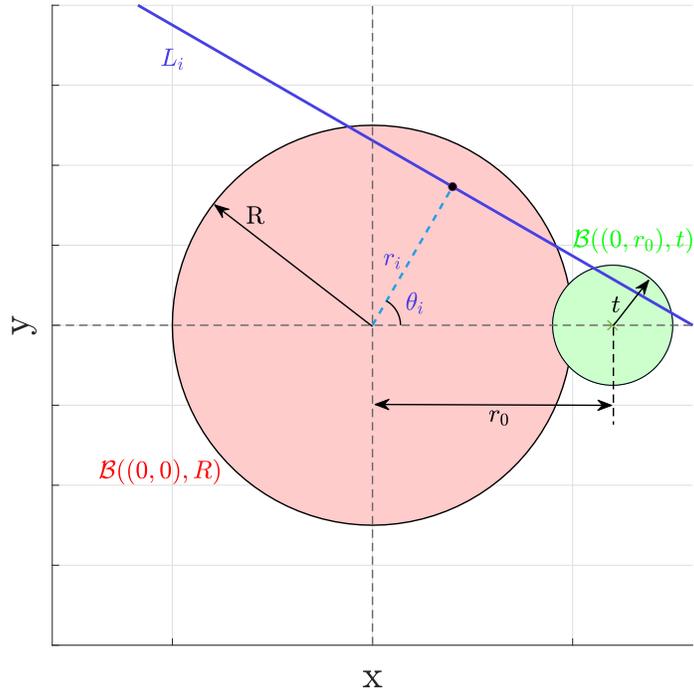}
    \caption{Illustration of the construction of a BLP and intersecting lines on $\mathcal{B}((0,r_0), t)$.}
    \label{fig:Construction}
\end{figure}
This is illustrated in Fig.~\ref{fig:Construction}, where the generating points of the lines are restricted to the disk $\mathcal{B}((0,0), R)$ (shown in red in the figure). It should be noted that the generating points do not form a BPP in $\mathcal{B}((0,0), R)$ but rather in $\mathcal{D}$\footnote[1]{The former would imply uniform location of points in the disk, while the latter correspond to uniform distances of the generating points between 0 and $R$.}. Furthermore, the \ac{BLP} is a non-stationary process, and unlike stationary point processes, the statistics of the \ac{BLP} cannot be characterized from the perspective of a single typical point located, say, at the origin. However, due to the isotropic construction of the \ac{BLP}, the properties of the \ac{BLP} as seen from a point depend only on its distance from the origin and not its orientation. Accordingly, without loss of generality, let us consider a test point located at $(0,r_0)$. First, we study the set of lines that intersect a disk of radius $t$ centered at the test point $(0,r_0)$, denoted by $\mathcal{B}((0,r_0), t)$ and shown in green in Fig.~\ref{fig:Construction}. This will then lead to the distance distribution of the nearest line as discussed in the next subsection.

\subsection{Domain Bands in $\mathcal{D}$ and the Distance Distribution to the Nearest \ac{BLP} Line}
Let us consider the line $L_i = \left\{(x, y) \in \mathbb{R}^2: x \cos(\theta_i) + y \sin(\theta_i) = r_i\right\}$, where $(\theta_i, r_i) \in \mathcal{D}$ and, the boundary of the disk $\mathcal{B}((0,r_0), t)$ is $(x - r_0)^2 + y^2 = t^2$. Solving these two equations simultaneously for $x$ and $y$,  we obtain the abscissa of the intersection points of $L_i$ and $\mathcal{B}((0,r_0), t)$ as:
\begin{align}
    x_i^* &= r_i\cos\theta_i + r_0\sin^2\theta_i \pm \sin\theta_i \sqrt{t^2 - (r_i - r_0\cos\theta_i)^2}.
    \label{eq:abscissa} 
\end{align}

In order to find the subset $\{(\theta_i, r_i)\} \subset \mathcal{D}$ for which the generated lines $\{L_i\}$ intersect $\mathcal{B}((0,r_0), t)$, we find the $(\theta_i, r_i)$ which results in $L_i$ being a tangent to $\mathcal{B}((0,r_0), t)$. For a given $\theta_i = \theta$, the value(s) of $r_i$ for which \eqref{eq:abscissa} has only one solution is obtained by solving $t^2 - (r_i - r_0\cos\theta_i)^2 = 0$, which results in
\begin{align}
    r_i^*(\theta) = r_0 \cos(\theta) \pm t.
    \label{eq:domain}
\end{align}

Let the two solutions above be represented by $r_1(\theta)$ and $r_2(\theta)$ for the positive and the negative signs, respectively. In other words, all the points $(\theta_i , r_i) \in \mathcal{D}$ that fall within the set bounded by the two curves of \eqref{eq:domain} generate lines $L_i$ in $\mathbb{R}^2$ that intersect $\mathcal{B}((0,r_0), t)$. Let us denote this set by $\mathcal{D}_{\rm B}(r_0, t)$.

\begin{figure}
\centering
    \subfloat[]
    {\includegraphics[width=0.4\textwidth]{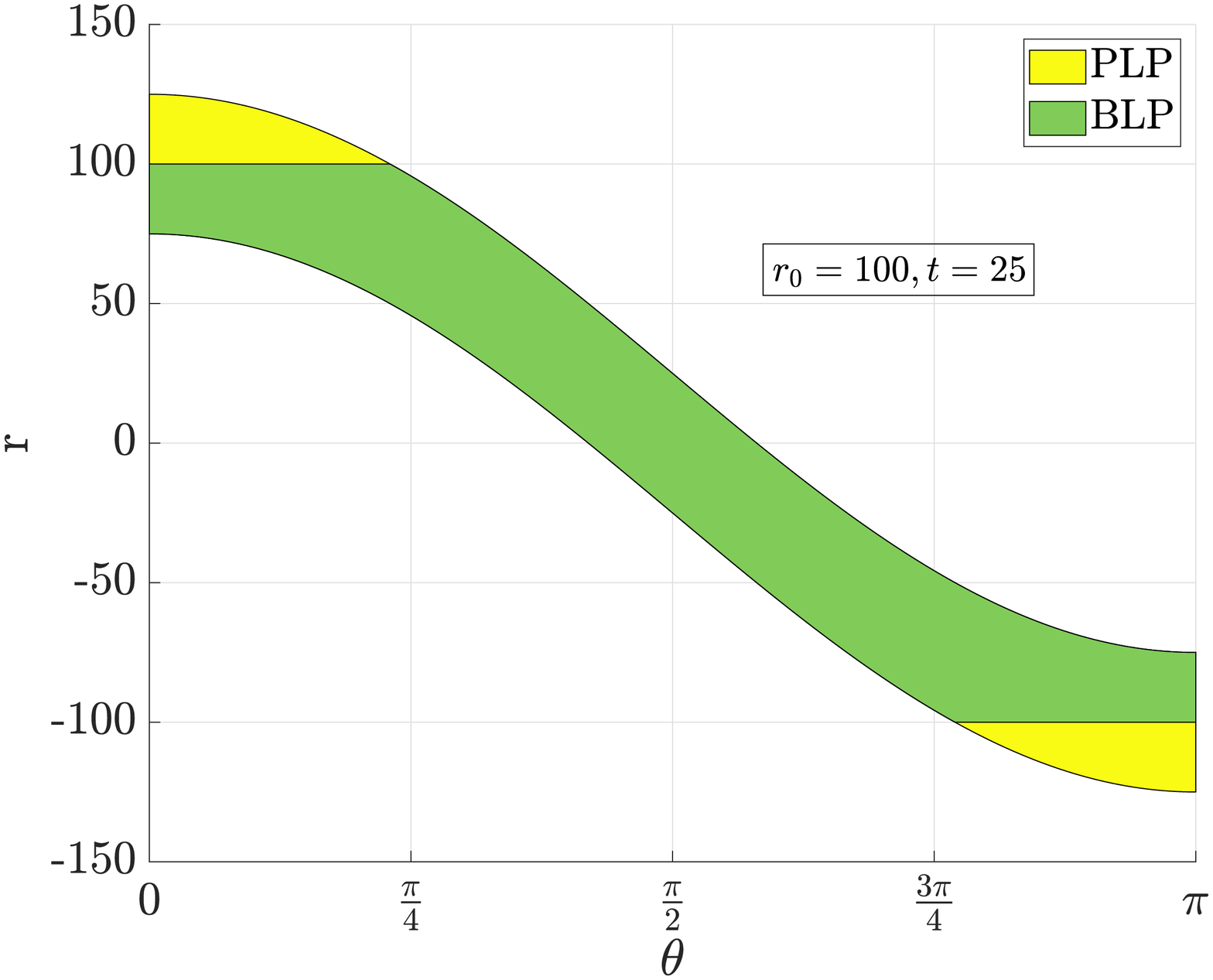}
    \label{fig:Domain0}}
    \hfil
    \subfloat[]
    {\includegraphics[width=0.4\textwidth]{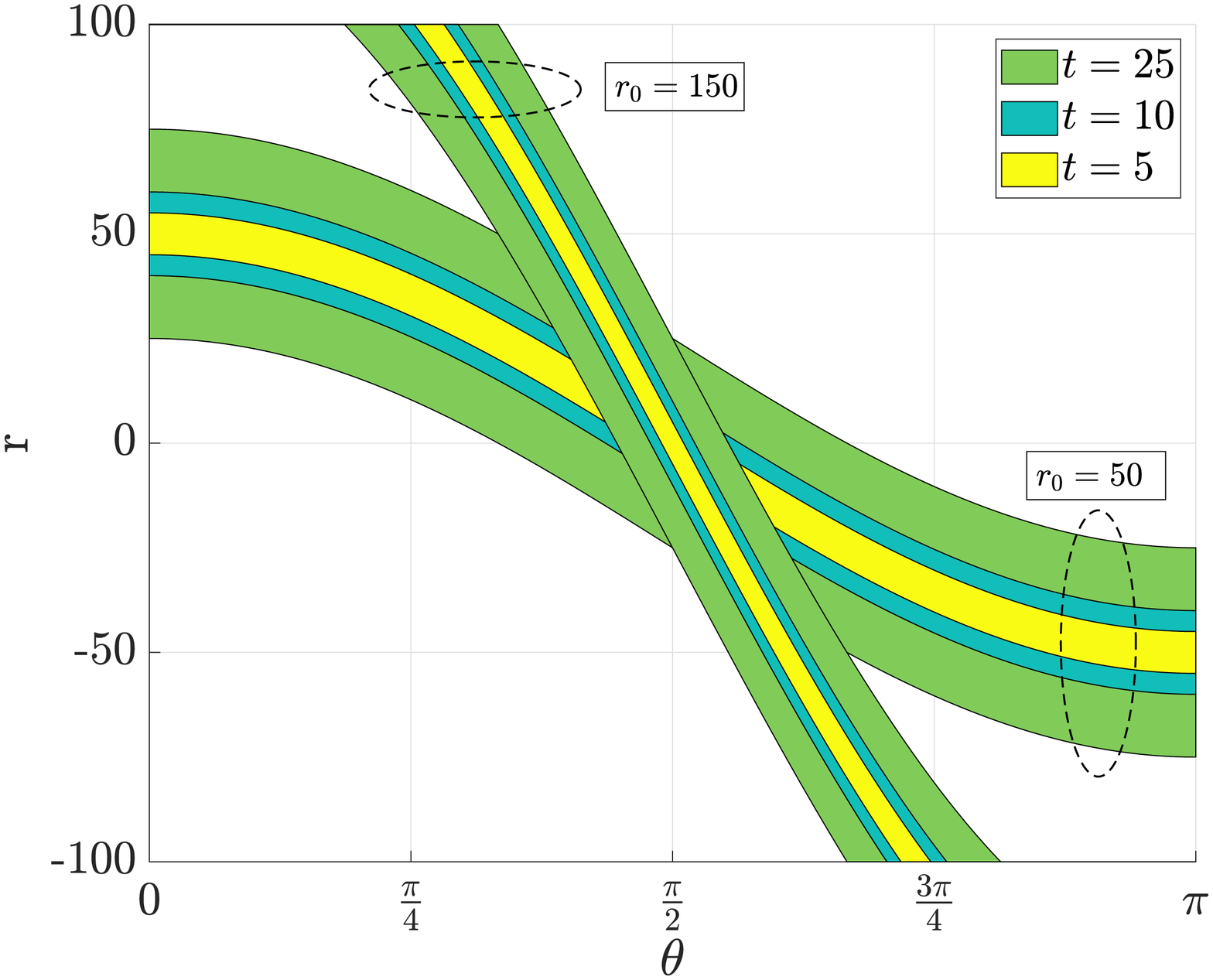}
    \label{fig:Domain1}}
\caption{(a) Illustration of the domain bands for a PLP and a BLP with $R = 100$. (b) Domain bands for different values of $t$ and $r_0$. Here $R = 100$. Note that when $r_0 + t \leq R$, the domain bands for PLP and BLP coincide.}
\label{fig:Domain} 
\end{figure}

In Fig.~\ref{fig:Domain0}, we plot $\mathcal{D}_{\rm B}(r_0, t)$, and due to their structure, $\mathcal{D}_{\rm B}(r_0, t)$ is referred to as {\it domain bands}. It is interesting to note that the domain band for the BLP is a clipped version of the PLP due to the restriction of the points to lie within $\mathcal{B}((0,0), R)$.
In Fig.~\ref{fig:Domain1}, we plot the domain bands for the BLP for different values of $r_0$ and $t$. Naturally, when $|r_0| + t < R$, the domain bands for \ac{BLP} and \ac{PLP} coincide. Additionally, the width of the band decreases as $t$ decreases or $r_0$ increases.

\begin{remark}
The area of the domain band for a PLP is $2 \pi t$ and is independent of $r_0$.
\end{remark}
The following classical result of the PLP follows.
\begin{corollary}{\bf \cite{6260478}}
The number of lines of a PLP intersecting $\mathcal{B}((0,r_0),t)$ is Poisson distributed with parameter $2\pi \lambda t$. Accordingly, the probability that no lines intersect $\mathcal{B}((0,r_0),t)$ is given by $\exp(- 2 \pi \lambda t)$.
\end{corollary}

However, in the case of a BLP, the area of the domain band needs to be computed by deriving the angles at which the domain bands are clipped due to its construction (contrary to the PLP where the domain bands are not clipped for any $\theta$)~\cite{ghatak22line}.

\begin{theorem}{\bf \cite{ghatak22line}} The area of the domain band $\mathcal{D}_{\rm B}(r_0, t)$ for a BLP is
\begin{align}
    &A_{\rm D}(r_0,t) = 
    \begin{cases}
        2 \pi t ; \quad  \hspace{7.6cm}  \text{for } r_0 + t \leq R \\
        2 \pi t - 2r_0 \sqrt{1 - \left(\frac{R - t}{r_0}\right)^{2}} + 2\left(R - t\right)  \cos^{-1} \left(\frac{R - t}{r_0}\right); \hspace{1cm} \text{for } r_0 + t > R \text{ and } r_0 - t \leq R \\
        2 \pi t - 2r_0 \left(\sqrt{1 - \left(\frac{R - t}{r_0}\right)^{2}} - \sqrt{1 - \left(\frac{R + t}{r_0}\right)^{2}} \right)  + 2\left(R - t\right) \cos^{-1} \left(\frac{R - t}{r_0}\right) \\
        \hspace{4.3cm} - 2\left(R + t\right)  \cos^{-1} \left(\frac{R + t}{r_0}\right); \hspace{0.5cm} \text{for } r_0 - t \geq R.
    \end{cases}
    \label{eq:areaofdomain}
\end{align}
\end{theorem}

\begin{corollary}{\bf \cite{ghatak22line}} {[Void Probability]}
  The probability that no line of the BLP intersects with $\mathcal{B}((0,r_0), t)$ is
\begin{align*}
    \mathcal{V}_{\rm BLP}(n_{\rm B}, \mathcal{B}((0,r_0), t)) = \left(\frac{2\pi R - A_{\rm D}(r_0,t)}{2\pi R}\right)^{n_{\rm B}},
\end{align*}
where, $n_{\rm B}$ is the number of lines of the BLP $\mathcal{L}$ and $A_{\rm D}$ is area of the domain bands.
\end{corollary}

\begin{corollary}{\bf \cite{ghatak22line}}
The CDF of the distance to the nearest line of the BLP from a test point at $(0,r_0)$ is,
\begin{align*}
    F_d(t) &= 1 - \mathcal{V}_{\rm BLP}(n_{\rm B}, \mathcal{B}((0,r_0), t)).
\end{align*} 
\end{corollary}

\subsection{Line Length Density and Measure}
Recall that one of the objectives of studying the BLP is to emulate different densities of streets in the city center and the suburbs. We characterize it using the {\it line length density} and \textit{line length measure}, defined as follows
\begin{definition}
    The line length measure is
    \begin{align*}
        \mathcal{R}(S) = n_{\rm B}\, \mathbb{E}\left(|L\cap S|_1\right), \quad S \subset R^2,
    \end{align*}
     where $|\cdot|_1$ is the Lebesgue measure in one dimension and $L$ is a line of the BLP. The corresponding radial density is
     \begin{align*}
         \rho(r)=\lim_{u\to 0} \frac{\mathcal{R} \big(\mathcal{B}((0,0),r+u) \setminus \mathcal{B}((0,0),r)\big)}{\pi \left(2u+u^2\right)}.
     \end{align*}
\end{definition}
The line length measure follows by integrating $\rho(r)$, i.e.,
\begin{align*}
    \mathcal{R} (S) = \int_{S} \rho(|{\bf x}|)\; \mathrm{d}{\bf x}, \quad S \subset \mathbb{R}^2.
\end{align*}

In order to study the {line length density}, we first determine the expected total length of chords in a disk, i.e., the line length measure.
\begin{theorem}
For a BLP generated by $n_{\rm B}$ lines within a disk of radius $R$,
\begin{align*}
    \rho (r) {=} \begin{cases} 
    \frac{n_{\rm B}}{2R},  &\text{if } r \leq R    \\
    \frac{n_{\rm B}}{\pi R} \arcsin{\left(\frac{R}{r}\right)}& \text{if }r > R.
    \end{cases}
\end{align*}
\begin{proof}
Let $S = \mathcal{B}\left((0,r_0), t\right)$. We have
\begin{equation}
    \mathcal{R}(S) = \mathbb{E} \left[K \bar{L}_1\right] = \bar{L}_1 \mathbb{E} \left[K \right] \overset{(a)}{=} \bar{L}_1 n_{\rm B} \left(\frac{A_{\rm D}(r_0,t)}{2\pi R}\right),
    \label{eq:llm}
\end{equation}
where $K$ is the number of lines intersecting disk $\mathcal{B}\left((0,r_0), t\right)$ and $\bar{L}_1$ is the expected length of the chord formed by a single line in the disk $S$. Step (a) follows from the expectation of the binomial distribution. For $\mathcal{B}\left((0,r_0), t\right)$, $\bar{L}_1$ is evaluated as 
\begin{align*}
    \bar{L}_1 = \frac{1}{A_{\rm D}(r_0,t)} \iint\limits_{A_{\rm D}} 2 \sqrt{\left(t^2 - (r_0 \cos(\theta) - r)^2)\right)}\, {\rm d}r {\rm d}\theta,
\end{align*}
where $A_{\rm D}(r_0,t)$ is obtained from \eqref{eq:areaofdomain}. 

\begin{figure}[t]
\centering
\subfloat[]
{\includegraphics[width=0.4\textwidth]{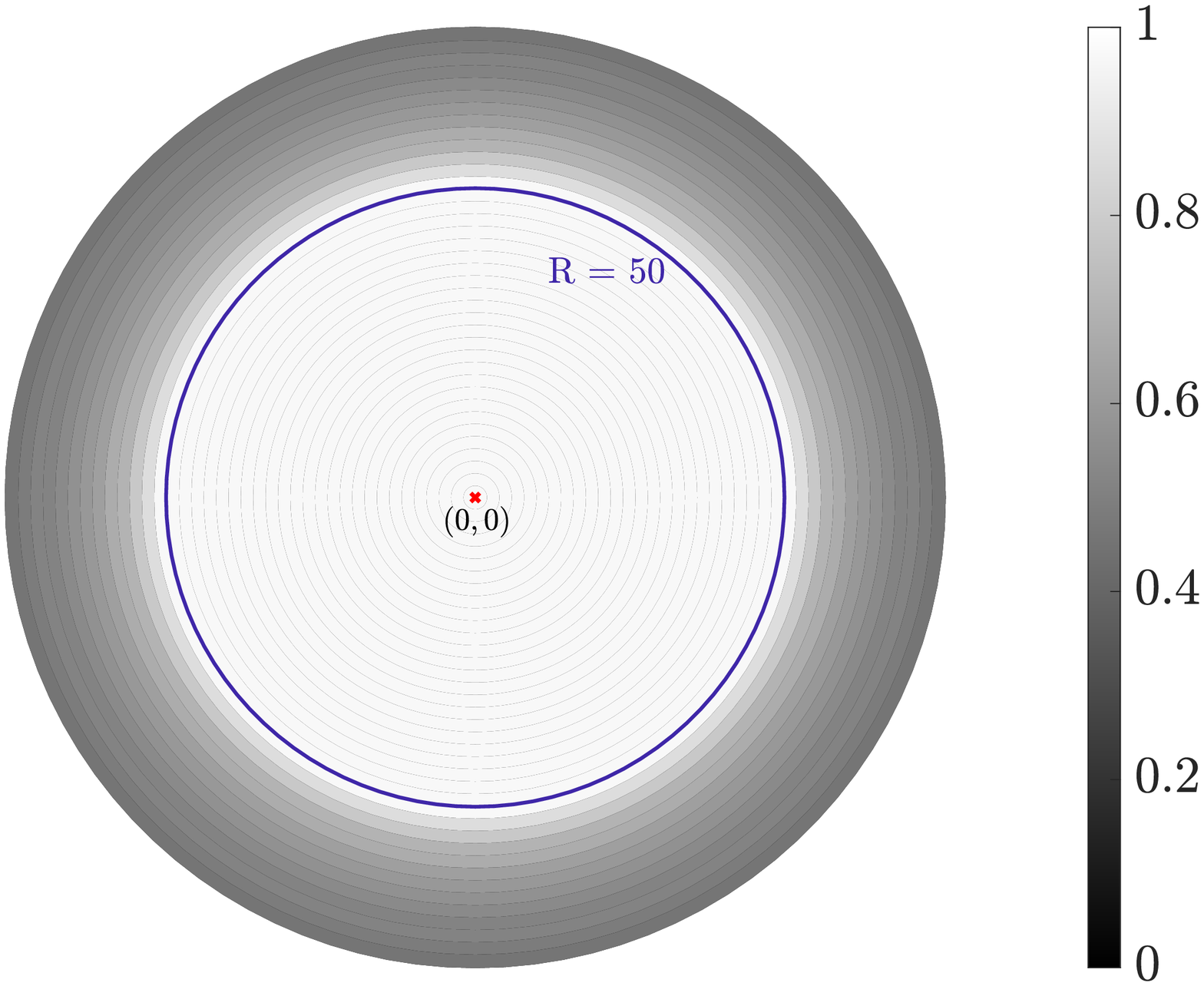}
\label{fig:annulis}}
\hfil
\subfloat[]
{\includegraphics[width=0.4\textwidth]{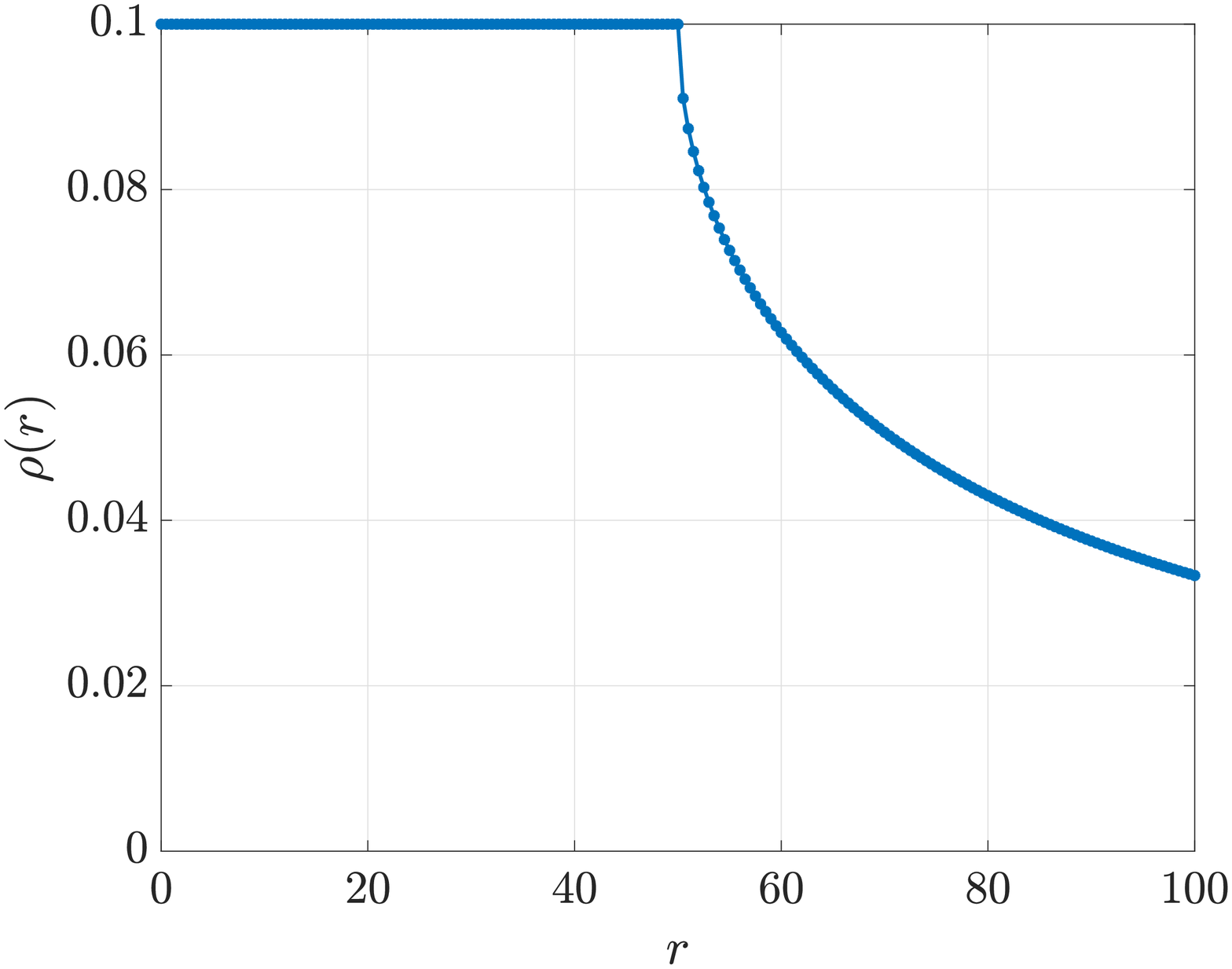}
\label{fig:lld}}
\caption{(a) Ratio of the line length measure to the area in concentric annuli of equal width $w = 2$. Here, $R=50$. (b) Line length density $\rho (r)$ for $R=50$ and $n_{\rm B} = 10$.}
\label{fig:lld2} 
\end{figure}

Next, let us consider concentric circles centered at the origin i.e. $r_0 = 0$ having radii $l = \{w, 2w, 3w, \ldots\}$, where $R$ is an integer multiple of $w$. The $i-${th} annulus is defined as the region between the circles of radii $iw$ and $(i+1)w$. Thus, the annuli generated due to these concentric circles have equal width $w$. Using (\ref{eq:llm}) we can find the ratio of the average length of line segments to the area, in the $i-$th annulus of width $w$ denoted by $\rho_{i}(w)$ as
\begin{align}
    &\rho_{i} (w) = 
    \begin{cases}
        \frac{n_{\rm B}}{2R} ; \quad  \hspace{10.2cm}  \text{for } (i+1)w \leq R, \\
        \frac{n_{\rm B}}{\pi w^2 (2i+1)} \Bigg(\left(\sqrt{(i+1)^2w^2 - R^2} + \arcsin\left({\frac{R}{(i+1)w}}\right)\times \frac{(i+1)^2w^2}{R}\right) \\
        \hspace{4.3cm} -\left(\sqrt{i^2w^2 - R^2} + \arcsin\left({\frac{R}{iw}}\right)\times \frac{i^2w^2}{R}\right) \Bigg); \hspace{0.5cm} \text{for } (i+1)w>R.
    \end{cases}
    \label{eq:annuluslength}
\end{align}
In Fig.~\ref{fig:annulis} we show a grey-scale plot for $\rho_i(2)$ for $R = 50$ and observe that the line length measure in the annuli decreases as $r_0$ increases. Leveraging this, we can characterize the line length density of the BLP as a limiting function of the density in annuli. Precisely, the statement of the theorem is obtained by substituting $iw = r$ and taking the limit $w \to 0$ in \eqref{eq:annuluslength}.
\end{proof}
\end{theorem}

The density $\rho(r)$ remains constant at $\frac{n_{\rm B}}{2R}$ for $r \leq R$ and then decreases as $\mathcal{O}(1/r)$ as $r\to\infty$. 

\subsection{Intersection Density}
In this section, we study the point process formed by the intersections of the lines of the \ac{BLP}. Accordingly, we introduce and characterize the intersection measure and the intersection density of the BLP.
\begin{theorem}
\label{theo:intersection}
The radial intersection density at a distance $r$ from the origin for a BLP generated by $n_{\rm B}$ lines within a disk of radius $R$ is
\begin{align*}
    \rho_{\times} (r) {=} \begin{cases} 
    \frac{n_{\rm B} (n_{\rm B} - 1)}{4 \pi R^2},  &\text{if } r \leq R,    \\
    \frac{n_{\rm B} (n_{\rm B} - 1)}{4\pi^2 R^2 r} \left(2r \arcsin{\left(\frac{R}{r}\right)} - \frac{2R}{r}\sqrt{r^2-R^2}\right) & \text{if } r > R.
    \end{cases}
\end{align*}
\end{theorem}
\begin{proof}
See Appendix~\ref{app:intersection}
\end{proof}
\begin{figure}[t]
\centering
\includegraphics[width = 0.45\textwidth]{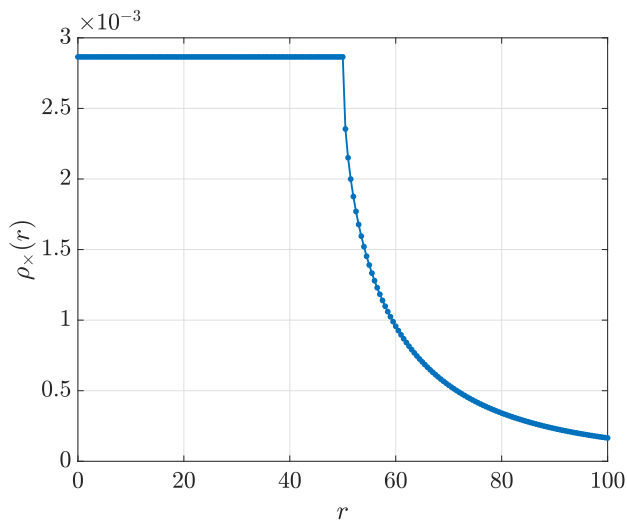}
\caption{Intersection density for $R=50$ and $n_{\rm B} = 10$.}
\label{fig:intersectionDensity}
\end{figure}

In Fig~\ref{fig:intersectionDensity} we see that intersection density first remains constant and then scales as $\mathcal{O}\left(\frac{1}{r}\right)$ as $r\to\infty$. By integrating the intersection density, we get the intersection measure
\begin{align*}
    \mathcal{R} (S) = \int_S \rho_{\times} (|{\bf x}|) \; \mathrm{d}{\bf x}, \quad S \subset \mathbb{R}^2.
\end{align*}
\begin{remark}
    The intersection measure in the $\mathbb{R}^2$ plane is
    \begin{align*}
        \mathcal{R}_{\times} = \int_0^{2\pi} \int_0^\infty \rho_{\times} (r)\; r\; \mathrm{d}r\; \mathrm{d}\theta = \binom{n_{\rm B}}{2} = \frac{n_{\rm B} (n_{\rm B} - 1)}{2},
    \end{align*}
    as expected.
\end{remark}

\subsection{Intersection Density of PLP}
In this section, we determine the intersection density and measure of PLP and compare it with that of the BLP.
\begin{theorem}
The intersection density for a PLP with density $\lambda_{\rm PPP}$  is
\begin{align*}
    \rho_{\rm P}(\lambda_{\rm PPP}) =  \pi\lambda_{\rm PPP}^2.
\end{align*}
\end{theorem}
\begin{proof}
Let $S = \mathcal{B}\left((0,0), t\right)$ and consider a PLP line $L_0$ generated at the point $(0,0)$ without loss of generality. The domain band $\mathcal{D}_{\rm PPP}$ corresponding to the intersection on $L_0$ depends only on the area in which the intersection density is calculated i.e., for a given $\theta$, the range of $r$ where if a line is generated intersects $(0,r_0)$ within $S$ is,
\begin{align*}
    \left(r_0\cos\theta - \sqrt{t^2 - r_0^2}\sin\theta\right) \leq r_i \leq \left(r_0 \cos\theta + \sqrt{t^2 - r_0^2} \sin\theta\right).
\end{align*}
Consequently, the area of the domain band averaged for uniformly distributed $r_0$ between 0 and $t$ can be written as
\begin{align*}
    A_{D_{\rm P}}(t) &= \frac{1}{t}\int_0^t\int_0^\pi \left(r_0\cos\theta + \sqrt{t^2 - r_0^2}\sin\theta\right)\; \mathrm{d}\theta\; \mathrm{d}r_0 - \frac{1}{t}\int_0^t\int_0^\pi \left(r_0 \cos\theta - \sqrt{t^2 - r_0^2} \sin\theta\right)\; \mathrm{d}\theta\; \mathrm{d}r_0 = \pi t.
\end{align*}
Accordingly, the probability that a line of the PLP intersects a single line within $S$ is
\begin{align*}
    \mathcal{P}_{\rm P} (t) = \frac{A_{D_{\rm P}}(t)}{2\pi t} =   \frac{1}{2}.
\end{align*}
Let us now suppose that $k$ lines are generated in $S$. With probability $\mathcal{P}_{\rm P} (t)$, each of them intersect line $L_0$. Thus, the average number of intersections on $L_0$ from the $k$ lines within $S$ is evaluated as
\begin{align*}
    \mathcal{N}^\prime = \sum_{j=0}^k j \binom{k}{j} \left(\mathcal{P}_{\rm P}(t) \right)^{j} \left(1 - \mathcal{P}_{\rm P}(t) \right)^{k-j} = \frac{k}{2}.
\end{align*}
Finally in order to determine the average number of intersections on all the lines within $S$, we take the expectation over the number of lines that are generated within $S$. This is evaluated as
\begin{align}
    \mathcal{N} &= \sum_{k=0}^\infty \frac{\left( 2\pi t\lambda_{\rm PPP}\right)^{(k+1)} e^{\left(-2\pi t\lambda_{\rm PPP} \right)}}{(k+1)!} \times (k+1) \times \frac{k}{2} \times \frac{1}{2}  = \frac{\left( 2\pi t\lambda_{\rm PPP}\right)^2}{4} = \pi^2 t^2\lambda^2_{\rm PPP}.
    \label{eq:intersectionPPP}
\end{align}
Next, similar to subsection 2.3, we consider concentric circles centered at the origin having radii $l = \{w, 2w, \dots \}$. As a result, the annuli formed by these concentric circles have the same width, $w$. In the $i-${th} annulus of width $w$, we use (\ref{eq:intersectionPPP}) to determine the ratio of the average number of intersections to the area to find the intersection density 
\begin{align}
    \rho_{{\rm P}, i} (w) &=  \frac{1}{\pi \left((i+1)w\right)^2 - \pi (iw)^2} \left(\left(\lambda_{\rm PPP} \pi (i+1)w\right)^2 - \left(\lambda_{\rm PPP} \pi iw\right)^2\right) = \pi\lambda_{\rm PPP}^2 .
    \label{eq:intersectionAnnulusP}
\end{align}
This implies that the density at any point at any distance from the origin is also constant. 
\end{proof}
As expected, the intersection density scales as $\lambda^2_{\rm PPP}$.

\subsection{Distance Distribution to the Nearest Intersection}
\begin{figure}[t]
    \centering
    \includegraphics[width = 0.45\textwidth]{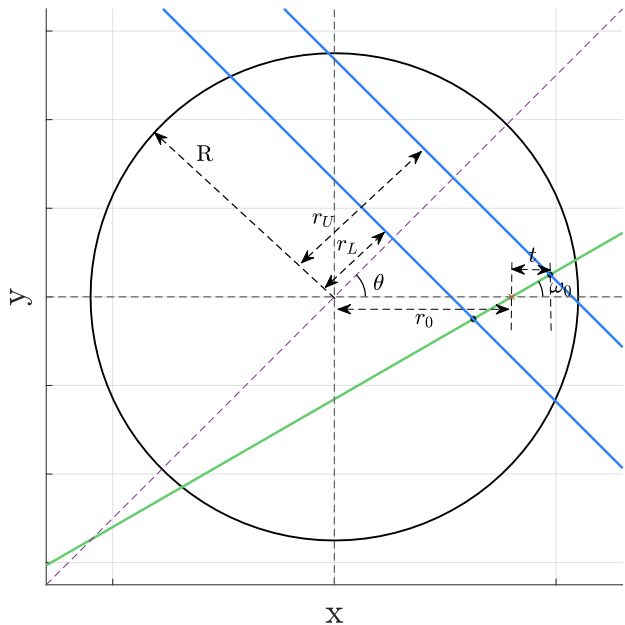}
    \caption{Illustration of distances $r_{\rm L}$ and $r_{\rm U}$ (denoted by blue lines respectively) for a line (denoted by the green line) passing through $(0, r_0)$ and having angle $\omega_0$ with the x-axis.}
    \label{fig:intersectionline}
\end{figure}
Let us consider a test point $(0, r_0)$ that lies on a line of the BLP. Consequently, $n_{\rm B} - 1$ lines of the BLP intersect the line containing the test point almost surely. Let $t$ be the distance to the nearest intersection from the test point. Also, let $\omega_0$ be the angle formed between the line passing through $(0, r_0)$ and the x-axis. This is illustrated in Fig. \ref{fig:intersectionline}, where for an intersection to be located at a distance of $t$ from the test point at an angle of $\omega_0$, there may exist a set of $r$ for a given $\theta$ wherein no lines should be generated. 
Accordingly, in order to find the distance distribution, we need the set of all such $(\theta_i,r_i)$ for which lines $L_i \in \mathcal{L}$  do not intersect the line passing through $(0, r_0)$ at an angle of $\omega_0$ within a distance of $t$ from the test-point. Now, for a given $\theta$, the range of $r$ where if a line is generated, intersects the test point within a distance $t$ is $[r_{\rm L}, r_{\rm U}]$ calculated as:

\begin{equation}
r_{\rm L} =
\begin{cases}
\max{\{-R,\min{\{R,(r_0\cos{(\theta)} - t\cos{(\theta-\omega_0)}\}}\}} & \text{if } \omega_0 \leq \frac{\pi}{2}\; \text{and}\; \theta\leq\omega_0+\frac{\pi}{2},
\\
\max{\{-R,\min{\{R,(r_0\cos{(\theta)} + t\cos{(\theta-\omega_0)}\}}\}} & \text{if } \omega_0 \leq \frac{\pi}{2}\; \text{and}\; \theta > \omega_0+\frac{\pi}{2},
\\
\min{\{R,\max{\{-R,r_0\cos{(\theta)} + t\cos{(\theta-\omega_0)}\}}\}} & \text{if } \omega_0 > \frac{\pi}{2}\; \text{and}\; \theta\leq\omega_0-\frac{\pi}{2},
\\
\max{\{-R,\min{\{R,(r_0\cos{(\theta)} - t\cos{(\theta-\omega_0)}\}}\}} & \text{if } \omega_0 > \frac{\pi}{2}\; \text{and}\; \theta > \omega_0-\frac{\pi}{2}.
\end{cases}
\label{eq:domainIL}
\end{equation}

\begin{equation}
r_{\rm U} =
\begin{cases}
\max{\{-R,\min{\{R,(r_0\cos{(\theta)} + t\cos{(\theta-\omega_0)}\}}\}} & \text{if } \omega_0 \leq \frac{\pi}{2}\; \text{and}\; \theta\leq\omega_0+\frac{\pi}{2},
\\
\max{\{-R,\min{\{R,(r_0\cos{(\theta)} - t\cos{(\theta-\omega_0)}\}}\}} & \text{if } \omega_0 \leq \frac{\pi}{2}\; \text{and}\; \theta > \omega_0+\frac{\pi}{2},
\\
\min{\{R,\max{\{-R,r_0\cos{(\theta)} - t\cos{(\theta-\omega_0)}\}}\}} & \text{if } \omega_0 > \frac{\pi}{2}\; \text{and}\; \theta\leq\omega_0-\frac{\pi}{2},
\\
\max{\{-R,\min{\{R,(r_0\cos{(\theta)} + t\cos{(\theta-\omega_0)}\}}\}} & \text{if } \omega_0 > \frac{\pi}{2}\; \text{and}\; \theta > \omega_0-\frac{\pi}{2}.
\end{cases}
\label{eq:domainIU}
\end{equation}

\begin{figure}[t]
    \centering
    {\includegraphics[width = 0.75\textwidth]{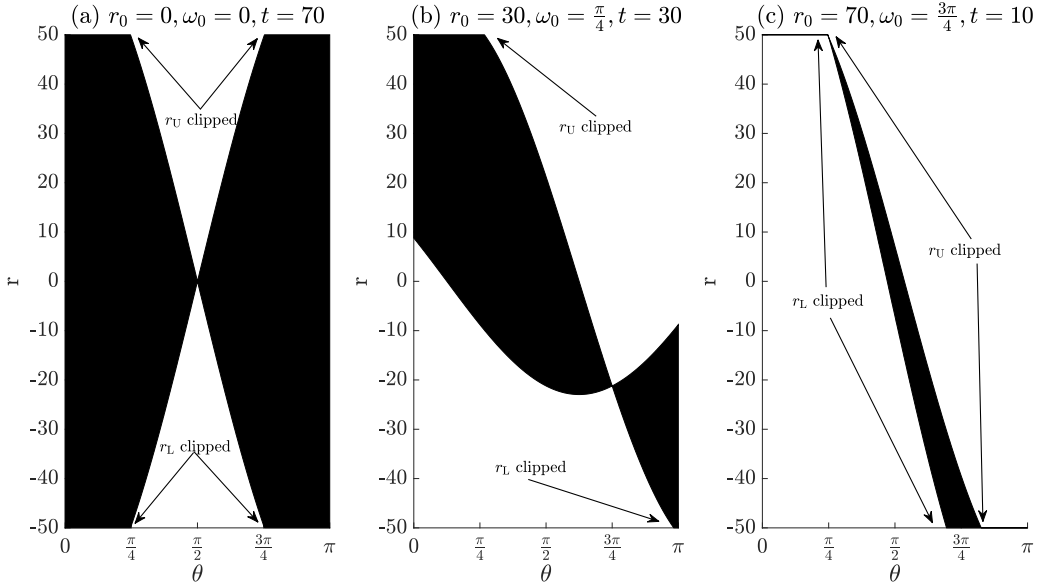}}
    \caption{Illustration of the region $D_{\rm I} (r_0, t)$ for different values of $r_0, \omega_0$ and $t$. Here $R=50$.}
    \label{fig:clipedI}
\end{figure}

The above equations of $r_{\rm L}$ and $r_{\rm U}$ are obtained using simple trigonometric calculations and include all possible cases for different values of $r_0, t, \theta$, and $\omega_0$. As an example, the case presented in Fig.~\ref{fig:intersectionline} corresponds to $\theta<\omega_0 + \frac{\pi}{2}$ and $\omega_0<\frac{\pi}{2}$, accordingly, we have $r_{\rm L} < r_{\rm U}$. We see that when $\theta>\omega_0 + \frac{\pi}{2}$ we can end up with scenarios $r_{\rm L} > r_{\rm U}$, thus equation \eqref{eq:domainIL} and \eqref{eq:domainIU} have been defined in such a way that $r_{\rm L}$ will always be less than $r_{\rm U}$ for all possible values of $r_0$ and $t$. Also, \eqref{eq:domainIL} and \eqref{eq:domainIU} consider the cases when $r_{\rm L}$ and $r_{\rm U}$ exceed $[R,-R]$.
\par Let $D_{\rm I}$ be the set of all such $(\theta, r)$ for which lines $L_i \in \mathcal{L}$ do not intersect the line passing through $(0, r_0)$ at angle of $\omega_0$ within a distance of $t$ from the test-point. Fig.~\ref{fig:clipedI} shows some examples of the domain band regions wherein a line should not be generated for it to not intersect the line passing through $(0, r_0)$ within a distance $t$ from the test point at an angle of $\omega_0$. In Fig \ref{fig:clipedI}(a), as $r_0 = 0, \omega_0=0$ and $t > R$ we see that $|r_0|+t>R$, the domain bands are getting clipped at 50 and -50 for most of the initial and final values of $\theta$. Likewise in Fig. \ref{fig:clipedI} (c) when the test point lies outside the circle of radius $R$ and $t = 10$ i.e., $|r_0|-t>R$ and, more values of $r_{\rm L}$ and $r_{\rm U}$ are clipped for $\theta<\frac{\pi}{2}$ and the total width of the band is also small, thus showcasing that test points lying outside and having small $t$ would experience fewer intersections.

\begin{corollary}{}
For a BLP line passing through $(0, r_0)$, the CDF of the distance $d_{\rm I}$ to the nearest intersection is
\begin{align*}
    F_{d_{\bf I}}(t) = 1 - \mathcal{V}_{B_{\bf I}} (r_0,t).
\end{align*}
\end{corollary}
The area of $D_{\rm I}$ for a BLP $\mathcal{L}$ defined on $[0,2\pi)\times [0,R]$ corresponding to $(0, r_0)$ and $t$ is
    $A_{D_{\rm I}}(r_0,t) = \int_0^\pi\int_0^\pi r_{\rm U}\; \mathrm{d}\theta\; \mathrm{d}\omega_0 - \int_0^\pi\int_0^\pi r_{\rm L}\; \mathrm{d}\theta\; \mathrm{d}\omega_0.$
Accordingly, the probability that no line of the BLP intersects the line passing through $(0, r_0)$ within a distance $t$ from the test point is given by
\begin{align*}
    \mathcal{V}_{B_{\rm I}} (r_0,t) = \left(1 - \frac{A_{D_{\rm I}}(r_0,t)}{2\pi R}\right)^{n_{\rm B}}. 
\end{align*}
\begin{figure}[t]
    \centering
    {\includegraphics[width = 0.5\textwidth]{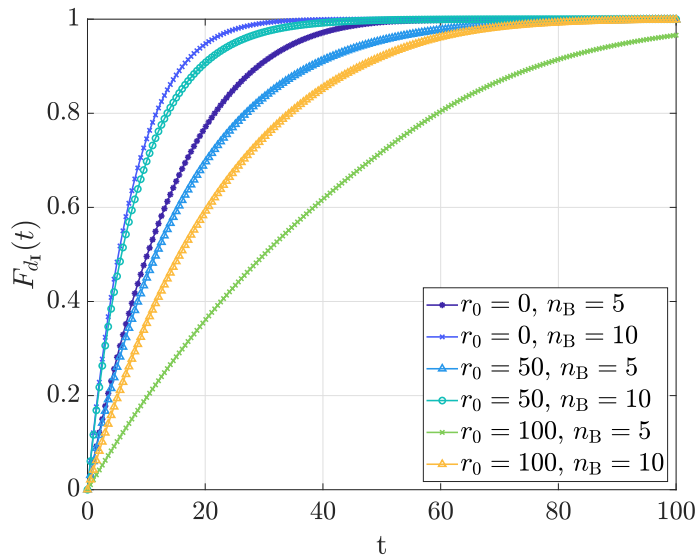}}
    \caption{Distance distribution to the nearest intersection from $(0, r_0)$.}
    \label{fig:cdfCliped}
\end{figure}
\\
In Fig. \ref{fig:cdfCliped} we plot the CDF for different values of $n_{\rm B}$ and $r_0$ for $R = 50$. We see that the nearest intersection is closer for higher $n_{\rm B}$ and smaller $r_0$ values.


\section{Binomial Line Cox Process}
\label{sec:BLCP}
On each line $L_i$ of $\mathcal{L}$, let us define an independent 1D \ac{PPP} $\Phi_i$ with intensity $\lambda$. A \ac{BLCP} $\Phi$, is the collection of all such points on all lines of the \ac{BLP}, i.e., $\Phi = \bigcup\limits_{i=1}^{n_{\rm B}} \Phi_i$. Thus, the \ac{BLCP} is a doubly-stochastic or Cox process of random points defined on random lines.

\subsection{Void Probability} 
Similar to the \ac{BLP}, we derive the void probability of the BLCP and consequently, its nearest point distance distribution. This has previously been reported in~\cite{ghatak22line}, however, compared to~\cite{ghatak22line}, we take a simpler approach for the derivation. In particular, leveraging the polar coordinates leads to much simpler trigonometric manipulations, and accordingly, we reduce the number of events for which void probability needs to be calculated explicitly.
\begin{theorem}
The probability that the disk $\mathcal{B}((0,r_0), t)$ contains no points of $\Phi$ is given by
\begin{align}
    \mathcal{V}_{\rm BLCP}\left(n_{\rm B}, \mathcal{B}((0,r_0), t)\right) &= \left[\frac{1}{2 \pi R} \int_0^{2\pi} \int_{r_0 \cos\theta - t}^{r_0 \cos\theta + t}  \exp\left({-\lambda C(\theta, r)}\right) \, \mathrm{d}r \, \mathrm{d}\theta \right]^{n_{\rm B}},
\end{align}
where,
\begin{align}
    C(\theta, r) &= 
            \begin{cases}
                2\sqrt{t^2 - (r_0 \cos \theta - r)^2}; & t \geq |r_0 \cos \theta - r|, \\
                0; & \text{otherwise},
            \end{cases}
\end{align}
is the length of the chord created by a line corresponding to $(\theta, r) \in \mathcal{D}$ in the disk $\mathcal{B}((0,r_0), t)$.
\end{theorem}

\begin{proof}
Let us first consider a line $L$ generated by a uniformly located point in $\mathcal{D}$ and the corresponding BLCP (which is simply a 1D PPP) defined on this line. Then the void probability for a single line case can be found as,
\begin{align}
     \mathcal{V}_{\rm BLCP}( 1, \mathcal{B}((0,r_0), t)) &= \frac{1}{2 \pi R} \int_0^{2\pi} \int_{r_0 \cos(\theta) - t}^{r_0 \cos(\theta) + t}  \exp\left({-\lambda C(\theta, r)}\right) \, \mathrm{d}r \, \mathrm{d}\theta.
\end{align}
The void probability is the probability that no BLP lines contain a BLCP point inside $\mathcal{B}((0,r_0), t)$. Let us recall that the distance of a line $L$ corresponding to $(\theta, r) \in \mathcal{D}$ from the test point $(0,r_0)$ is, $|r_0 \cos(\theta) - r|$. Accordingly, if $t < |r_0 \cos(\theta) - r|$, then $L$ does not intersect $\mathcal{B}((0,r_0), t)$ resulting in a chord length 0. Then, given that $C(\theta, r) \geq 0$, the probability that no points of the BLCP lie on that chord is given as $\exp\left(-\lambda C\right)$ following the void probability of a PPP with intensity $\lambda$. Due to the independence of the locations of the $n_{\rm B}$ generating points in $\mathcal{D}$ and the independence of the individual 1D PPPs $\Phi_i$, the void probability can be evaluated as
\begin{align}
    \mathcal{V}_{\rm BLCP}\left(n_{\rm B}, \mathcal{B}((0,r_0), t)\right) = &\frac{1}{(2 \pi R)^{n_{\rm B}}} \prod_{L \in \mathcal{L}}\left[\int_0^{2\pi} \int_{r_0 \cos\theta - t}^{r_0 \cos\theta + t}  \exp\left({-\lambda C(r, \theta)}\right) \, \mathrm{d}r \, \mathrm{d}\theta \right] \nonumber \\
    =& \left[\frac{1}{2 \pi R} \int_0^{2\pi} \int_{r_0 \cos(\theta) - t}^{r_0 \cos(\theta) + t}  \exp\left({-\lambda C(\theta, r)}\right) \, \mathrm{d}r \, \mathrm{d}\theta \right]^{n_{\rm B}}.
\end{align}
\end{proof}

\begin{corollary}{[Distance Distribution]}
\label{cor:distdist}
Following the void probability, the distance distribution of the nearest BLCP point from the test point $(0,r_0)$ is
\begin{align}
    F_{d_1}(t) = 1 - \mathcal{V}_{\rm BLCP}\left(n_{\rm B}, \mathcal{B}((0,r_0), t)\right). \nonumber 
\end{align}
\end{corollary}
From a wireless network perspective where the locations of the \ac{AP}s are modeled as points of a BLCP, the above result characterizes the distance distribution to the nearest \ac{AP}. This will be used in deriving the communication performance metrics in Section~\ref{sec:app}.

\subsection{Palm Perspective of the \ac{BLCP}}
Next, we study the \ac{BLCP} from the perspective of a point of the process itself, using Palm calculus\footnote[2]{In point process theory, the Palm probability refers to the probability measure conditioned on a point of the process being at a certain location~\cite{haenggi2012stochastic}.}.
Let us recall that for a PLCP $\Phi_{\rm PLCP}$ with $\lambda$ as the density of the points on the lines, we have $\mathbb{P}(\Phi_{\rm PLCP} \in Y \mid o) = \mathbb{P}(\Phi_{\rm PLCP} \cup \Phi_0 \cup \{o\} \in Y)$, where $\Phi_0$ is a 1D PPP on a line that passes through the origin. In other words, the Palm distribution, i.e., conditioning on a point of the PLCP $\Phi_{\rm PLCP}$ to be at the origin, is equivalent to \textit{adding} (i) an independent Poisson process of intensity $\lambda$ on a line through the origin with uniform independent angle and (ii) an atom at the origin to the PLCP. Similarly, for a BLCP, conditioning on a point to be located at ${\bf x}$ is equivalent to considering an atom at ${\bf x}$, a 1D PPP on a line passing through ${\bf x}$ and a BLCP $\Phi^!$ defined on a BLP consisting of $n_{\rm B} - 1$ lines in the same domain. Thus, the Palm measure of the BLCP can be expressed as follows.
\begin{lemma}
For a BLCP $\Phi$ defined on a BLP $\mathcal{P}$ with $n_{\rm B}$ lines, we have

\begin{align}
\mathbb{P}\left(\Phi \in Y \mid {\bf x} \in \Phi\right) = \mathbb{P}\left(\Phi^! \cup \Phi_{\bf x} \cup \{\bf x\} \in Y\right),
\end{align}

where $\Phi_{\bf x}$ is a 1D PPP on a randomly oriented line that passes through ${\bf x}$.
\end{lemma}
The Palm perspective is necessary to characterize the statistics of the performance metrics conditioned on an event that a node, e.g., a transmitter or an \ac{AP} is located at a given point. For example, consider a network in which the location of the \ac{AP}s are modeled according to a BLCP. If a receiver located at the origin associates with a transmitter located at ${\bf x}$, then the interfering \ac{AP}s are located not only in the other $n_{\rm B} - 1$ lines but also on a line necessarily passing through ${\bf x}$. The applications of this concept will be evident in the next section where we employ the derived framework to analyze a wireless communication network. Prior to that, let us derive the \ac{PGFL} of the shifted and reduced point process by conditioning on the location of the nearest point from the origin.

\subsection{Probability Generating Functional}
Here, we characterize the \ac{PGFL} of the BLCP $\Phi$. The \ac{PGFL} of a point process $\Phi$ evaluated for a function $\nu$ is defined mathematically as the Laplace functional of $- \log \nu$ and is calculated as $\mathbb{E} \left[\prod\limits_{{\bf x} \in \Phi} \nu({\bf x})\right]$. In this paper, we are interested in isotropic functions that depend only on the distance of the points from the origin, i.e., we consider functions of the form $f(||{\bf x}||)$. 
\begin{definition}
Let ${\bf x}_1$ be the nearest point of a BLCP from $(0, r_0)$. Then, the shifted and reduced point process is defined as $\Phi' = \Phi - (0, r_0) \backslash \{{\bf x}_1\}$.
\end{definition}
The motivation for studying the properties of $\Phi'$ in the context of wireless networks is as follows; if the \ac{AP} locations are modeled as a BLCP, then $\Phi'$ represents the locations of the interfering \acp{AP} from the perspective of a user located at a distance $r_0$ from the origin and connected to an \ac{AP} located at a distance $||{\bf x}_1||$ from the user. The next theorem characterizes the PGFL for the shifted and reduced BLCP $\Phi'$.

\begin{theorem}
\label{theo:pgfl}
For a shifted and reduced BLCP $\Phi' = \Phi - (0, r_0) \backslash \{{\bf x}_1\}$ defined on a BLP with $n_{\rm B}$ lines generated within $\mathcal{B}((0,0), R)$, the PGFL of a function $f(r) = f(||{\bf x}||),$ conditioned on $d_1 = ||{\bf x}_1||$ is given as
    \begin{multline*}
    G(r_0, f(\cdot)) = \frac{1}{A_{\rm D}(r_0,d_1)} \iint\limits_{\mathcal{D}_{\rm B}(0, d_1)} \exp\left(-2\lambda\int_{\sqrt{d_1^2 - ( r_0 \cos\theta - r)^2}}^{\infty}\!1- f\left(\sqrt{y^2 + (r_0\cos{\theta}-r)^2}\,\right)\:{\rm d}y\right){\rm d}r {\rm d}\theta \times\\
    \left(\frac{1}{2\pi R}\right)^{n_{\rm B}-1} \Bigg(\;\,\iint\limits_{\mathcal{D}_{\rm B}(0, d_1)}  \exp\left(-2\lambda\int_{\sqrt{d_1^2 - ( r_0 \cos\theta - r)^2}}^{\infty}\!1- f\left(\sqrt{y^2 + (r_0\cos{\theta}-r)^2}\,\right)\:{\rm d}y\right) {\rm d}r {\rm d}\theta + \\ \iint\limits_{\mathcal{D} \backslash \mathcal{D}_{\rm B}(0, d_1)} \exp\left(-2\lambda\int_{0}^{\infty}\!1- f\left(\sqrt{y^2 + (r_0\cos{\theta}-r)^2}\,\right)\:{\rm d}y\right){\rm d}r {\rm d}\theta \Bigg)^{n_{\rm B}-1} .
    \end{multline*}
    where $d_1 = ||{\bf x}_1||$ is the distance to the nearest point of $\Phi - (0, r_0)$ from the origin. Consequently, the PGFL of $\Phi'$ is evaluated as $\mathbb{E}_{d_1}\left[G(r_0, f(\cdot))\right]$, where the distribution of $d_1$ is given by Corollary~\ref{cor:distdist}.
\end{theorem}

\begin{proof}
Let $f(r)$ be a positive, measurable, monotonic, and bounded function for the first part of the proof. Here we will find the PGFL of the point process $\Phi'\cap \mathcal{B}((0,0), t)$. In other words, the shifted and reduced point process of interest is restricted to the disk centered at the origin with radius $t$. The theorem follows from the monotone convergence theorem with $t \to \infty$. Recall that for a PPP of intensity $\lambda$ and a function $\nu(\cdot)$, the PGFL is~\cite{stoyan2013stochastic}:
\begin{align}
    G_{\text{PPP}} = \exp\left(-\lambda \int_{\mathbb{R}^d} \left(1 - \nu({\bf x})\right) {\rm d}{\bf x}\right).
    \label{eq:PPPPGFL}
\end{align}
Next, note that the distance of a BLP line $L_i$ corresponding to the generating point $(r, \theta)$ in $\mathcal{D}$ from the origin in $\mathbb{R}^2$ is $|r_0 \cos \theta - r|$. A point located at a distance $y$ from the perpendicular projection of $(0,r_0)$ to $L_i$, has a distance $\sqrt{y^2 + (r_0 \cos \theta - r)^2}$ from $(0,r_0)$. The length of the chord is $2\sqrt{t^2 - (r_0 \cos \theta - r)^2}$ when $t \geq |r_0 \cos \theta - r|$.
\begin{align*}
    G_1(r_0,r,\theta) &= \lim_{t\to \infty} \exp\left(-2\lambda\int_{0}^{\sqrt{t^2 - (r_0 \cos \theta - r)^2}}\!1- f\left(\sqrt{y^2 + (r_0\cos{\theta}-r)^2}\,\right)\:dy\right) \nonumber \\
    &\overset{(a)}{=} \exp\left(-2\lambda\int_{0}^{\infty}\!1- f\left(\sqrt{y^2 + (r_0\cos{\theta}-r)^2}\,\right)\:dy\right),
\end{align*}
where step (a) is due to the monotone convergence theorem.
In $\Phi'$, each line can be either of two types (a) intersecting with $\mathcal{B}((0,0), d_1)$ and (b) non-intersecting with $\mathcal{B}((0,0), d_1)$, where $d_1$ is the distance from the origin to the nearest point of $\Phi- (0, r_0)$. For a particular $r_0$ and $d_1$, a line is intersecting if $|r_0\cos{\theta} - r| \geq d_1$ and non-intersecting otherwise. Thus, we can write the conditional PGFL of the intersecting and non-intersecting lines after averaging over $(r, \theta) \in \mathcal{D}$ as
 \begin{equation}\label{eqivsni}
  \begin{split}
    G_{\rm I}(r_0,d_1) &= \frac{1}{A_{\rm D}(r_0,d_1)} \iint\limits_{\mathcal{D}_{\rm B}(0, d_1)} \exp\left(-2\lambda\int_{\sqrt{d_1^2 - ( r_0 \cos\theta - r)^2}}^{\infty}\!1- f\left(\sqrt{y^2 + (r_0\cos{\theta}-r)^2}\,\right)\:dy\right) {\rm d}r {\rm d}\theta, \\
    G_{\rm NI}(r_0,d_1) &= \frac{1}{(2\pi R - A_{\rm D}(r_0, d_1))} \iint\limits_{\mathcal{D} \backslash \mathcal{D}_{\rm B}(0, d_1)} \exp\left(-2\lambda\int_{0}^{\infty}\!1- f\left(\sqrt{y^2 + (r_0\cos{\theta}-r)^2}\,\right)\:dy\right) {\rm d}r {\rm d}\theta.
  \end{split}
 \end{equation}
 
Next, note that the line containing the nearest point (at a distance $d_1$) of the BLCP intersects the disk $\mathcal{B}((0,0), d_1)$ almost surely. Whereas, the other $n_{\rm B} - 1$ lines may or may not intersect the disk depending on their generating point. Accordingly, the PGFL for $\Phi - (0, r_0)$ is evaluated as
 \begin{align}
    G(r_0,f(\cdot)) &\overset{(a)}{=}
    \underbrace{G_{\rm I}(r_0,d_1)}_{{\rm T}_1} \sum_{n=0}^{n_{\rm B}-1} \binom{n_{\rm B}-1}{n} \Biggl[\underbrace{\left(\frac{A_{\rm D}(r_0, d_1)}{2 \pi R}\times G_{\rm I}(r_0,d_1)\right)^n}_{{\rm T}_2} \nonumber \\
    & \hspace*{6cm} \times \underbrace{\left(\left(1 - \frac{A_{\rm D}(r_0, d_1)}{2 \pi R}\right)\times G_{\rm NI}(r_0,d_1)\right)^{n_{\rm B} - n -1}}_{{\rm T}_3}\Biggr] \nonumber \\
    &\overset{(b)}{=} G_{\rm I}(r_0,d_1) \left(\frac{1}{2\pi R}\right)^{n_{\rm B}-1}  \left(A_{\rm D}(r_0, d_1) \; G_{\rm I}(r_0, d_1) + \left({2\pi R - A_{\rm D}(r_0, d_1)}\right) \; G_{\rm NI}(r_0, d_1)\right)^{n_{\rm B}  - 1}. \label{eq:cpgfl}
\end{align} 
In step (a), the term ${\rm T}_1$ corresponds to the line containing the nearest point (recall the Palm perspective discussed in the previous sub-section). The term ${\rm T}_2$ corresponds to the probability that a set of $n$ lines intersect the disk and the conditional PGFL given that the lines intersect the disk. The term ${\rm T}_3$ corresponds to the probability that a set of $n_{\rm B} - n - 1$ lines do not intersect the disk and the conditional PGFL given that the lines do not intersect the disk. The statement of the theorem follows from the above.
\end{proof}


\section{Application - Transmission Success Probability in Wireless Networks}
\label{sec:app}
In wireless networks, several performance metrics are studied using the {\it transmission success probability}. It is the \ac{CCDF} of the \ac{SINR} over the fading coefficients and the spatial process governing the locations of the \ac{AP}s. In this section, first we define this metric and then we characterize it using the results derived in the previous sections.

\subsection{Success Probability - Definition}
Let $\Phi$ be a point process (not necessarily a BLCP) consisting of points $\{{\bf x}_i\} \subset \mathbb{R}^2$, $i = 1, 2, \ldots$. Consider a separate test point located at the origin. For convenience, let us assume that the points of $\Phi$ are ordered according to their distance from the origin, i.e., $||{\bf x_1}|| \leq ||{\bf x_2}|| \leq \ldots$. If the points of $\Phi$ emulate the locations of the \ac{AP}s of a wireless network relative to a receiver located at the origin, the receiver connects to the \ac{AP} located at ${\bf x_1}$. This is known as the nearest-AP association. 

Each wireless link experiences fluctuations of the received power due to constructive and destructive superposition of multiple reflecting paths in the propagation environment. This is termed small-scale fading. Classically, the impact of this is taken into account by multiplying the received signal with a random variable $h$ with exponential distribution with parameter 1~\cite{haenggi2009stochastic}. For a path-loss exponent $\alpha$, the \ac{SINR} $\xi (r_0)$ is
\begin{align}
    \xi (r_0) = \frac{\xi_0||{\bf x_1}||^{-\alpha}h_1}{1 + \xi_0\sum_{{\bf x} \in \Phi\backslash \{{\bf x_1}\}} ||{\bf x}||^{-\alpha}h_{\bf x}},
\end{align}
where $\xi_0$ is a constant that takes into account the transmit power, AWGN noise, path-loss constant, as well as the transmit and receive antenna gains. We assume that this parameter is the same for each transmit node. Typically, each $h_{\bf x}$ is independent of each other and identically distributed~\cite{haenggi2009stochastic}. For the ease of notation, let us represent $||{\bf x_i}||$ by $d_{\bf x}$. Now, the transmission success probability at a threshold of $\gamma$ is defined as the CCDF of $\xi (r_0)$~\cite{haenggi2009stochastic}:
\begin{align}
    p_{\rm S}(\gamma) = \mathbb{P}[\xi (r_0) > \gamma].
\end{align}
This represents the probability that an attempted transmission by the nearest \ac{AP} located at ${\bf x_1}$ is decoded successfully by the receiver at the origin. In what follows, we refer to the transmission success probability as {\it success probability}.

\subsection{Success probability for BLCP Locations of \acp{AP}}
The \ac{BLCP} is a relevant model for studying deployment locations of \ac{AP}s along the streets of a city, or e.g., along alleyways of industrial warehouses. As discussed before, the analysis of the network performance depends on the location of the test point\footnote{It may be noted that the wireless network performance analyzed at the test point referred here corresponds to the performance evaluated at the typical point at a location ${\bf x}$ of a stationary receiver point process.}. However, since the BLP is isotropic, we may infer that its properties as seen from a point only depend on its distance from the center and not its orientation. Accordingly, the test point can be considered to be located along the x-axis, i.e., we analyze the performance from the perspective of a test point located at $(0,r_0)$, without loss of generality. Equivalently, we can consider the receiver at the origin and study the statistics of the shifted point process $\Phi - (0, r_0)$. Let us assume that the receiver establishes a connection with its nearest \ac{AP} (i.e., the \ac{AP} located at the nearest BLCP point from the receiver), consequently experiencing interference from all other APs. In such a case, the success probability is characterized by the following result.
\begin{theorem}
For the network where locations of the \acp{AP} are modeled as BLCP, the success probability for a receiver located at $(0, r_0)$ is given by
\begin{align*}
    p_{\rm S}(\gamma) = \mathbb{E}_{d_1} \left[\exp \left( \frac{-\gamma}{\xi_0 {d_1}^{-\alpha}} \right) G\left(r_0, \frac{1}{1 + \frac{\gamma r^{-\alpha}} {d_1^{-\alpha}}} \right)\right],
\end{align*}
where $G(r_0, f(\cdot))$ is given by Theorem~\ref{theo:pgfl}.
\end{theorem}
\begin{proof}
The success probability can be evaluated as
\begin{align}
    p_{\rm S}(\gamma) &= \mathbb{P}[\xi (r_0) > \gamma] = \mathbb{P}\left[\frac{\xi_0 {d_1}^{-\alpha} h_1} {1 +  \xi_0 \sum_{{\bf x} \in \Phi^{\prime}} {d_{\bf x}}^{-\alpha} {h_{\bf x}}} > \gamma \right] = \mathbb{P}\left[h_1 > \frac{\gamma \xi_0 \sum_{{\bf x} \in \Phi^{\prime}} {d_{\bf x}}^{-\alpha} {h_{\bf x}} + \gamma} {\xi_0 d_1^{-\alpha}}\right] \nonumber \\
    & =  \mathbb{E} \left[\exp\left(\frac{-\gamma \xi_0 \sum_{{\bf x} \in \Phi^{\prime}} {d_{\bf x}}^{-\alpha} {h_{\bf x}} - \gamma} {\xi_0 d_1^{-\alpha}}\right)\right]
    = \mathbb{E}_{d_1} \left[\exp \left(\frac{-\gamma}{\xi_0 d_1^{-\alpha}} \right) \mathbb{E}^!_{{\bf x}_1,h_{\bf x}} \left[ \exp \left(\frac{-\gamma  \sum_{{\bf x} \in \Phi'}d_{\bf x}^{-\alpha} h_{\bf x}}{d_1^{-\alpha}}\right) \right] \right]. \label{eq:intermediate}
\end{align}
Here, $\mathbb{E}^!_{{\bf x}_1}$ refers to the expectation taken with respect to the Palm probability of the shifted and reduced point process, i.e., conditioned on a point of $\Phi - (0,r_0)$ being located at ${\bf x_1}$ and then removing it. The 1st term, $\exp\left(\frac{-\gamma}{\xi_0 d_1^{-\alpha}} \right)$, takes into account the impact of the noise and thus only depends on $d_1$ and $N_0$. The second term $\mathbb{E}^!_{{\bf x}_1,h_{\bf x}} \left[ \cdot\right]$, takes into account the impact of the interference and can be further simplified
\begin{align*}
    \mathbb{E}^!_{{\bf x}_1, h_{\bf x}} &\left[ \exp \left(\frac{-\gamma \sum_{{\bf x} \in \Phi^{\prime}} d_{\bf x}^{-\alpha} h_{\bf x}}{d_1^{-\alpha}}\right)\right] = \mathbb{E}^!_{{\bf x}_1} \Bigg[\mathbb{E}_{h_1} \Bigg[\prod_{{\bf x} \in \Phi^{\prime}} \exp\bigg(\frac{-\gamma d_{\bf x}^{-\alpha} h_{\bf x}} {d_1^{-\alpha}}\bigg)\Bigg]\Bigg]\\
    & = \mathbb{E}^!_{{\bf x}_1} \left[\prod_{{\bf x} \in \Phi^{\prime}} \left[\mathbb{E}_{h_1} \exp \left(\frac{-\gamma d_{\bf x}^{-\alpha} h_{\bf x}} {d_1^{-\alpha}}\right)\right]\right]= \mathbb{E}^!_{{\bf x}_1} \left[\prod_{{\bf x} \in \Phi^{\prime}} \frac{1}{1 + \frac{\gamma d_{\bf x}^{-\alpha}}{d_1^{-\alpha}}} \right] = G\left(r_0, \frac{1}{1 + \frac{\gamma r^{-\alpha}} {d_1^{-\alpha}}}\right) \\
    &= \frac{1}{A_{\rm D}(r_0, d_1)} \iint\limits_{ \mathcal{D}_{\rm B}(0, d_1)} \exp\left(-2\lambda\int_{\sqrt{d_1^2 - ( r_0 \cos\theta - r)^2}}^{\infty}\!\left(\frac{\gamma \left[y^2 + (r_0 \cos\theta - r)^2 \right]^{-\frac{\alpha}{2}}} {d_1^{-\alpha} + \gamma \left[y^2 + (r_0 \cos\theta - r)^2 \right]^{-\frac{\alpha}{2}}} \right)\:{\rm d}y\right){\rm d}r {\rm d}\theta \\
    &\left(\frac{1}{2\pi R}\right)^{n_{\rm B}-1} \Bigg(\iint\limits_{ \mathcal{D}_{\rm B}(0, d_1)}  \exp\left(-2\lambda\int_{\sqrt{d_1^2 - ( r_0 \cos\theta - r)^2}}^{\infty}\!\left(\frac{\gamma \left[y^2 + (r_0 \cos\theta - r)^2 \right]^{-\frac{\alpha}{2}}} {d_1^{-\alpha} + \gamma \left[y^2 + (r_0 \cos\theta - r)^2 \right]^{-\frac{\alpha}{2}}} \right)\:{\rm d}y\right){\rm d}r {\rm d}\theta +\\ 
    &\iint\limits_{ \mathcal{D} \backslash \mathcal{D}_{\rm B}(0, d_1)} \exp\left(-2\lambda\int_{0}^{\infty}\!\left(\frac{\gamma \left[y^2 + (r_0 \cos\theta - r)^2 \right]^{-\frac{\alpha}{2}}} {d_1^{-\alpha} + \gamma \left[y^2 + (r_0 \cos\theta - r)^2 \right]^{-\frac{\alpha}{2}}} \right)\:{\rm d}y\right){\rm d}r {\rm d}\theta\Bigg)^{n_{\rm B}-1}.
\end{align*}
Employing the above in \eqref{eq:intermediate} completes the proof.
\end{proof}

\begin{corollary}{}
 For $\alpha=2$ and $n_{\rm B} = 1$, the PGFL of $\Phi'$ for the function $\frac{1}{1 + \frac{\gamma r^{-2}} {d_1^{-2}}}$ is given as
    \begin{multline*}
        G\left(r_0, \frac{1}{1 + \frac{\gamma r^{-2}} {d_1^{-2}}}\right) = \frac{1}{A_{\rm D}(r_0, d_1)} \iint\limits_{ \mathcal{D}_{\rm B}(0, d_1)} \exp\Biggl(-2\lambda \frac{\gamma d_1^2}{\sqrt{\gamma d_1^2 +(r_0 \cos\theta - r)^2}} \\
        \times\arctan{\left(\sqrt{\frac{\gamma d_1^2 + (r_0 \cos\theta - r)^2}{d_1^2-(r_0 \cos\theta - r)^2}}\right)} \Bigg){\rm d}r {\rm d}\theta.
    \end{multline*}
\end{corollary}

One important metric of interest, especially for broadband applications, is the complementary CDF of the data rate. Let us recall that according to Shannon's formula, the instantaneous channel capacity $C$ is $C = W\log_2(1 + \xi)$. Now, for a packet deadline of $T$ and a file size $b$ for low-latency communications, the communication rate (in terms of bits per second) should be sustained at $C_0 = \frac{b}{T}$ or higher.
\begin{corollary}{}
The probability that the instantaneous channel capacity $C$ is $C_0$ or higher can be evaluated using the success probability as
    \begin{align*}
    \mathbb{P}\left(C > C_0\right) &= \mathbb{P}\left(W\log(1+\xi (r_0)) > C_0\right) = \mathbb{P}\left(\xi (r_0) > 2^{\frac{C_0}{W}} - 1\right) = \mathbb{P}\left(\xi (r_0) > 2^{\frac{b}{TW}} - 1\right) = p_{\rm S}(2^{\frac{b}{TW}} - 1) .
    \end{align*} 
\end{corollary}


\section{Application - Meta Distribution of the SINR in BLCP} 
\label{sec:meta}
Although the success probability is a useful metric for planning a wireless network and tuning network parameters, it only provides an average view of the network across all possible realizations of $\Phi$. This inhibits the derivation of a fine-grained view into the network. In this regard, the meta distribution, i.e., the distribution of the success probability conditioned on $\Phi$ provides a framework to study the same~\cite{haenggi2021meta1,haenggi2021meta2}. The conditional success probability, i.e., $P_s(\gamma) = \mathbb{P}(\xi (r_0) \geq \gamma | \Phi)$ is a random variable due to the random $\Phi$. It is precisely this random variable whose statistics we want to study. Its CCDF, called the meta distribution of the SINR, is given as
\begin{align}
    \mathcal{P}_{\rm M}(\gamma,\beta) = \mathbb{P}\left(P_s(\gamma) \geq \beta\right) = \mathbb{P}\left(\mathbb{P}(\xi (r_0) \geq \gamma \mid \Phi) \geq \beta\right).
\end{align}
which is a function of two parameters $\gamma \geq 0$ and $0 \leq \beta \leq 1$. In addition, consider another important aspect of wireless networks --- not all transmitters transmit simultaneously but are controlled by an {\it access scheme}. In particular, let us assume a simple ALOHA access scheme wherein, when the connected \ac{AP} transmits, each interfering \ac{AP} transmits with a probability $p$~\cite{abramson1970aloha}. Thus, each interference term is weighted by the probability of the corresponding node transmitting. Let the set of locations of the interfering nodes be denoted by $\mathcal{C} \subset \Phi'$. In this scheme, the conditional success probability can be obtained similarly to the success probability as
\begin{align*}
    P_s(\gamma) &= \mathbb{P} \left(\xi (r_0) \geq \gamma\; |\; \Phi\right) = \mathbb{P}\left[\frac{\xi_0 {d_1}^{-\alpha} h_1} {1 +  \xi_0 \sum_{{\bf x} \in \Phi^{\prime}} {h}_{\bf x} d_{\bf x}^{-\alpha}\textbf{1}({\bf x} \in \mathcal{C})} \geq \gamma\; |\; \Phi \right]\\
    &\hspace*{-0.7cm}= \mathbb{P}\left[h_1 > \frac{\gamma + \gamma \xi_0 \sum_{{\bf x} \in \Phi^{\prime}} {h}_{\bf x} d_{\bf x}^{-\alpha} \textbf{1}({\bf x} \in \mathcal{C})} {\xi_0 d_1^{-\alpha}}\right]
    \overset{(a)}{=} \mathbb{E}_{h_{\bf x}} \left[\exp\left(\frac{-\gamma - \gamma \xi_0 \sum_{{\bf x} \in \Phi^{\prime}} {h}_{\bf x} d_{\bf x}^{-\alpha} \textbf{1}({\bf x} \in \mathcal{C})} {\xi_0 d_1^{-\alpha}}\right)\right]\\
    &\hspace*{-0.7cm}= \exp{\left(\frac{-\gamma}{\xi_0 d_1^{-\alpha}}\right)} \mathbb{E}_{h_{\bf x}} \left[\exp\left(\frac{-\gamma \xi_0 \sum_{{\bf x} \in \Phi^{\prime}} {h}_{\bf x} d_{\bf x}^{-\alpha} \textbf{1}({\bf x} \in \mathcal{C})} {\xi_0 d_1^{-\alpha}}\right)\right]\\
    &\hspace*{-0.7cm}= \exp{\left(\frac{-\gamma}{\xi_0 d_1^{-\alpha}}\right)} \left(\prod_{{\bf x} \in \Phi^{\prime}} p \mathbb{E}_{h_{\bf x}} \exp\left(\frac{-\gamma \xi_0 d_{\bf x}^{-\alpha} {h}_{z}} {\xi_0 d_1^{-\alpha}}\right) + 1 - p\right)
    \overset{(b)}{=} \exp{\left(\frac{-\gamma}{\xi_0 d_1^{-\alpha}}\right)} \left( \prod_{{\bf x} \in \Phi^{\prime}} \frac{p}{1+\frac{\gamma d_{\bf x}^{-\alpha}}{d_1^{-\alpha}}} +1 - p \right).
\end{align*}
Step (a) is due to the exponential distribution of $h_1$. Step (b) follows from the Laplace transform of the exponentially distributed $h_{\bf x}$. In general, directly deriving the distribution of the random variable $P_s(\gamma)$ is most likely impossible. The standard approach to circumvent this challenge is by first deriving its moments and then transform them to the distribution~\cite{wang2022hausdorff}.
\begin{theorem}
\label{theo:mdmoments}
The $b-$th moment of $P_s(\gamma)$ conditioned on $d_1$ for any $b \in \mathbb{C}$ is given as
\begin{align*}
   M_b(d_1) &= \exp{\left(\frac{-b\gamma}{\xi_0 d_1^{-\alpha}}\right)} \frac{1}{A_{\rm D}(r_0, d_1)} \iint\limits_{\mathcal{D}_{\rm B}(0, d_1)} \exp\left(-2\lambda\int_{\sqrt{d_1^2 - l^2}}^{\infty}\! 1 - \left(\frac{pd_1^{-\alpha}} {d_1^{-\alpha} + \gamma \left[y^2 + l^2 \right]^{-\frac{\alpha}{2}}} +1-p\right)^b\:{\rm d}y\right) \nonumber \\
   & {\rm d}r {\rm d}\theta \times \left(\frac{1}{2\pi R}\right)^{n_{\rm B}-1} \Bigg(\;\,\iint\limits_{ \mathcal{D}_{\rm B}(0, d_1)}  \exp\left(-2\lambda\int_{\sqrt{d_1^2 - l^2}}^{\infty}\!1 - \left(\frac{pd_1^{-\alpha}} {d_1^{-\alpha} + \gamma \left[y^2 + l^2 \right]^{-\frac{\alpha}{2}}} +1-p\right)^b\:dy\right) \nonumber \\ 
   &{\rm d}r {\rm d}\theta + \iint\limits_{\mathcal{D} \backslash  \mathcal{D}_{\rm B}(0, d_1)} \exp\left(-2\lambda\int_{0}^{\infty}\!1 - \left(\frac{pd_1^{-\alpha}} {d_1^{-\alpha} + \gamma \left[y^2 + l^2 \right]^{-\frac{\alpha}{2}}} +1-p\right)^b\:{\rm d}y\right){\rm d}r {\rm d}\theta\Bigg)^{n_{\rm B}-1},
\end{align*}
where $l=r_0 \cos\theta - r$. Taking an expectation over $d_1$ (see Corollary~\ref{cor:distdist}) results in the unconditioned $b-$th moment.
\end{theorem}
\begin{proof}
We have
\begin{align}
    M_b &= \mathbb{E}^!_{{\bf x}_1} \left[ \left(\exp{\left(\frac{-\gamma}{\xi_0 d_1^{-\alpha}}\right)} \prod_{{\bf x} \in \Phi^{\prime}} \frac{p}{1+\frac{\gamma d_{\bf x}^{-\alpha}}{d_1^{-\alpha}}} +1 - p\right)^b\right] \stackrel{(c)}{=} \exp{\left(\frac{-b\gamma}{\xi_0 d_1^{-\alpha}}\right)} \mathbb{E}^!_{{\bf x}_1} \left[ \prod_{{\bf x} \in \Phi^{\prime}} \left(\frac{p}{1+\frac{\gamma d_{\bf x}^{-\alpha}}{d_1^{-\alpha}}} +1 - p\right)^b \right] \nonumber \\
    &\overset{(d)}{=} \exp{\left(\frac{-b\gamma}{\xi_0 d_1^{-\alpha}}\right)} \frac{1}{A_{\rm D}(r_0, d_1)} \iint\limits_{ \mathcal{D}_{\rm B}(0, d_1)} \exp\left(-2\lambda\int_{\sqrt{d_1^2 - l^2}}^{\infty}\! 1 - \left(\frac{pd_1^{-\alpha}} {d_1^{-\alpha} + \gamma \left[y^2 + l^2 \right]^{-\frac{\alpha}{2}}} +1-p\right)^b\:{\rm d}y\right) \nonumber \\
    & {\rm d}r {\rm d}\theta \times \left(\frac{1}{2\pi R}\right)^{n_{\rm B}-1} \Bigg(\;\,\iint\limits_{ \mathcal{D}_{\rm B}(0, d_1)}  \exp\left(-2\lambda\int_{\sqrt{d_1^2 - l^2}}^{\infty}\!1 - \left(\frac{pd_1^{-\alpha}} {d_1^{-\alpha} + \gamma \left[y^2 + l^2 \right]^{-\frac{\alpha}{2}}} +1-p\right)^b\:dy\right){\rm d}r {\rm d}\theta + \nonumber \\ 
    &\iint\limits_{\mathcal{D} \backslash  \mathcal{D}_{\rm B}(0, d_1)} \exp\left(-2\lambda\int_{0}^{\infty}\!1 - \left(\frac{pd_1^{-\alpha}} {d_1^{-\alpha} + \gamma \left[y^2 + l^2 \right]^{-\frac{\alpha}{2}}} +1-p\right)^b\:{\rm d}y\right){\rm d}r {\rm d}\theta\Bigg)^{n_{\rm B}-1}.
    \label{eq:moment}
\end{align}
\end{proof}
Step (c) follows because the expectation is over $\Phi^{\prime}$, thus the first term comes out of expectation. Step (d) follows from the PGFL of the BLCP. Consequently, the meta distribution of the SINR can be calculated using the Gil-Palaez theorem as~\cite{md_haenggi}
\begin{align*}
    F_{P_{\rm s}}(z) = \frac{1}{2} + \frac{1}{\pi}\int_0^\infty \frac{\Im(e^{-ju \log(z)}M_{ju})}{u} \mathrm{d}u,
\end{align*}
where $\Im\left(\cdot\right)$ is the imaginary part of the argument and $j^2 = -1$.

\begin{corollary}
For $\alpha = 2$, the first moment of $P_s$ conditioned on $d_1$ can be simplified as
\begin{align*}
     p_S(\gamma) &= M_1 = \exp{\left(\frac{-\gamma}{\xi_0 d_1^{-\alpha}}\right)}\times \frac{1}{A_{\rm D}(r_0,d_1)} \iint\limits_{ \mathcal{D}_{\rm B}(0, d_1)} \exp\left(-2\lambda \frac{p\gamma d_1^2}{\sqrt{\gamma d_1^2 +l^2}}\times\arctan{\left(\sqrt{\frac{\gamma d_1^2 + l^2}{d_1^2-l^2}}\right)} \right) \\
    &{\rm d}r {\rm d}\theta \times\left(\frac{1}{2\pi R}\right)^{n_{\rm B}-1} \Biggl(\;\,\iint\limits_{ \mathcal{D}_{\rm B}(0, d_1)}  \exp\left(-2\lambda\frac{p\gamma d_1^2}{\sqrt{\gamma d_1^2 +l^2}}\times\arctan{\left(\sqrt{\frac{\gamma d_1^2 + l^2}{d_1^2-l^2}}\right)}\right){\rm d}r {\rm d}\theta \\
    & \hspace*{5cm}+ \iint\limits_{\mathcal{D} \backslash  \mathcal{D}_{\rm B}(0, d_1)} \exp\left(-\lambda \frac{p\pi\gamma d_1^2}{\sqrt{\gamma d_1^2 +l^2}}\right){\rm d}r {\rm d}\theta\Biggr)^{n_{\rm B}-1},
\end{align*}
where $l=r_0 \cos\theta - r$.
\end{corollary}

\par The capacity of a wireless network can be studied in terms of the successful transmission density $p\lambda p_S(\gamma)$. This represents the density of the simultaneous transmitters that attempt and experience successful transmissions. 

Furthermore, the mean local delay is the expected number of transmissions required for successful transmission. It is obtained by setting $b=-1$~\cite{md_haenggi}. Naturally, in the case of latency-constrained networks, this is a critical parameter. 

\begin{corollary}
 For $\alpha = 2$, the mean local delay conditioned on $d_1$ is
\begin{align*}
    \mathcal{D}(p) = \frac{1}{p} M_{-1},
\end{align*}
where
\begin{multline*}
    M_{-1} = \exp{\left(\frac{\gamma}{\xi_0 d_1^{-\alpha}}\right)}\times \frac{1}{A_{\rm D}(r_0, d_1)} \iint\limits_{ \mathcal{D}_{\rm B}(0, d_1)} \exp\left(2\lambda \frac{p\gamma d_1^2}{\sqrt{\gamma d_1^2(1-p) +l^2}}\times \arctan{\left(\sqrt{\frac{\gamma d_1^2(1-p) + l^2}{d_1^2-l^2}}\right)} \right) \\
    {\rm d}r {\rm d}\theta \times \left(\frac{1}{2\pi R}\right)^{n_{\rm B}-1} \Bigg(\;\,\iint\limits_{ \mathcal{D}_{\rm B}(0, d_1)}  \exp\left(2\lambda\frac{p\gamma d_1^2}{\sqrt{\gamma d_1^2(1-p) +l^2}}\times \arctan{\left(\sqrt{\frac{\gamma d_1^2(1-p) + l^2}{d_1^2-l^2}}\right)}\right){\rm d}r {\rm d}\theta +\\ \iint\limits_{\mathcal{D} \backslash  \mathcal{D}_{\rm B}(0, d_1)} \exp\left(\lambda \frac{p\pi\gamma d_1^2}{\sqrt{\gamma d_1^2(1-p) +l^2}}\right){\rm d}r {\rm d}\theta\Bigg)^{n_{\rm B}-1}.
\end{multline*}
where $l=r_0 \cos\theta - r$.
\end{corollary}
In the next section, we discuss how due to the impact of interference, optimizing the channel access probability $p$ for minimizing the mean local delay is non-trivial.


\section{Numerical Results and Discussion}
\label{sec:NRD}
In this section, we discuss some numerical results to highlight the applications of the derived framework in analyzing the wireless network. Unless otherwise stated, all results are for $R = 50, \alpha=2, \xi_0 = 2.9858\cdot 10^{-8}$ and, $\gamma=0.1$.

\begin{figure}[t]
    \centering
    \includegraphics[width = 0.5\textwidth]{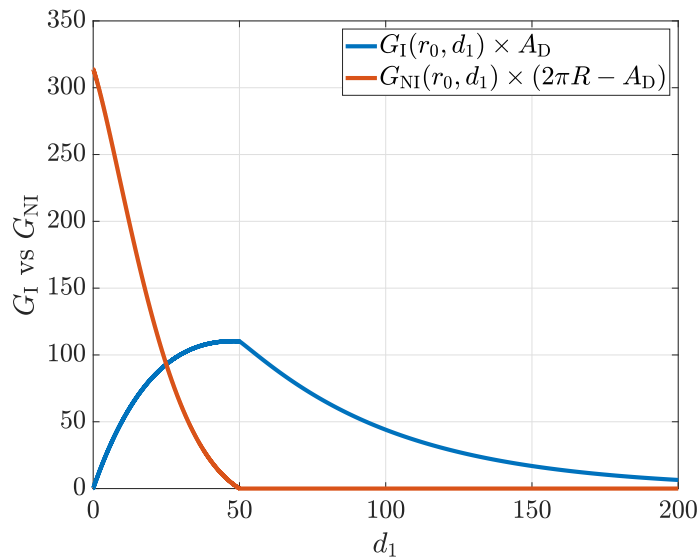}
    \caption{Conditional PGFL of intersecting and non-intersecting lines (see \eqref{eqivsni}). Here $r_0 = 0$, $R = 50$, $\lambda = 0.1$ and $n_{\rm B} = 10$. Recall that under Palm, one line intersects almost surely.}
    \label{fig:ivsni}
\end{figure}
\subsection{On the Success Probability}
First, let us observe some features of the weighted conditional PGFLs $A_{\rm D}(r_0, d_1) G_{\rm I}(r_0, d_1)$ and \newline $\left(2\pi R - A_{\rm D}(r_0, d_1)\right) \cdot G_{\rm NI}(r_0, d_1)$ for the function $f(r) = \frac{1}{1 + \gamma (d_1/r)^\alpha}$ as derived in \eqref{eqivsni}). From \eqref{eq:cpgfl} note that these conditional \acp{PGFL} constitute the overall functional $G(r_0, f(\cdot))$, i.e., 

\begin{align}
G(r_0, f(\cdot)) =  G_{\rm I}(r_0,d_1) \left(\frac{1}{2\pi R}\right)^{n_{\rm B}-1}  \left(A_{\rm D}(r_0, d_1) \; G_{\rm I}(r_0, d_1) + \left({2\pi R - A_{\rm D}(r_0, d_1)}\right) \; G_{\rm NI}(r_0, d_1)\right)^{n_{\rm B}  - 1}.
\nonumber 
\end{align}

Moreover, since $p_{\rm S}$ is proportional to $G(r_0, f(\cdot))$ for a given $d_1$, it is important to study the trends in the weighted conditional PGFLs with $d_1$. For lower values of $d_1$, $\mathcal{D}_{\rm B}(0, d_1)$ is smaller than $\mathcal{D} \backslash \mathcal{D}_{\rm B}(0, d_1)$. Accordingly, Fig.~\ref{fig:ivsni} shows that $(2\pi R - A_{\rm D}(0, d_1))G_{\rm NI}(0, d_1)$ has a value $2\pi R = 314.15$ at $d_1 = 0$ and decreases with $d_1$ until it becomes zero exactly at $d_1=50$. This is due to the fact that at $d_1 = R$, all lines intersect the disk. On the contrary, $A_{\rm D}(0, d_1)G_{\rm I}(0, d_1)$ has a value 0 at $d_1 = 0$ and increases till $d_1=50$, as the area corresponding to intersecting lines increases. Beyond $d_1=50$, $A_{\rm D}(0, d_1)G_{\rm I}(0, d_1)$ decreases precisely due to the increasing distance of the serving \ac{AP} from the receiver.



\begin{figure}[t]
\centering
\subfloat[]
{\includegraphics[width=0.4\textwidth]{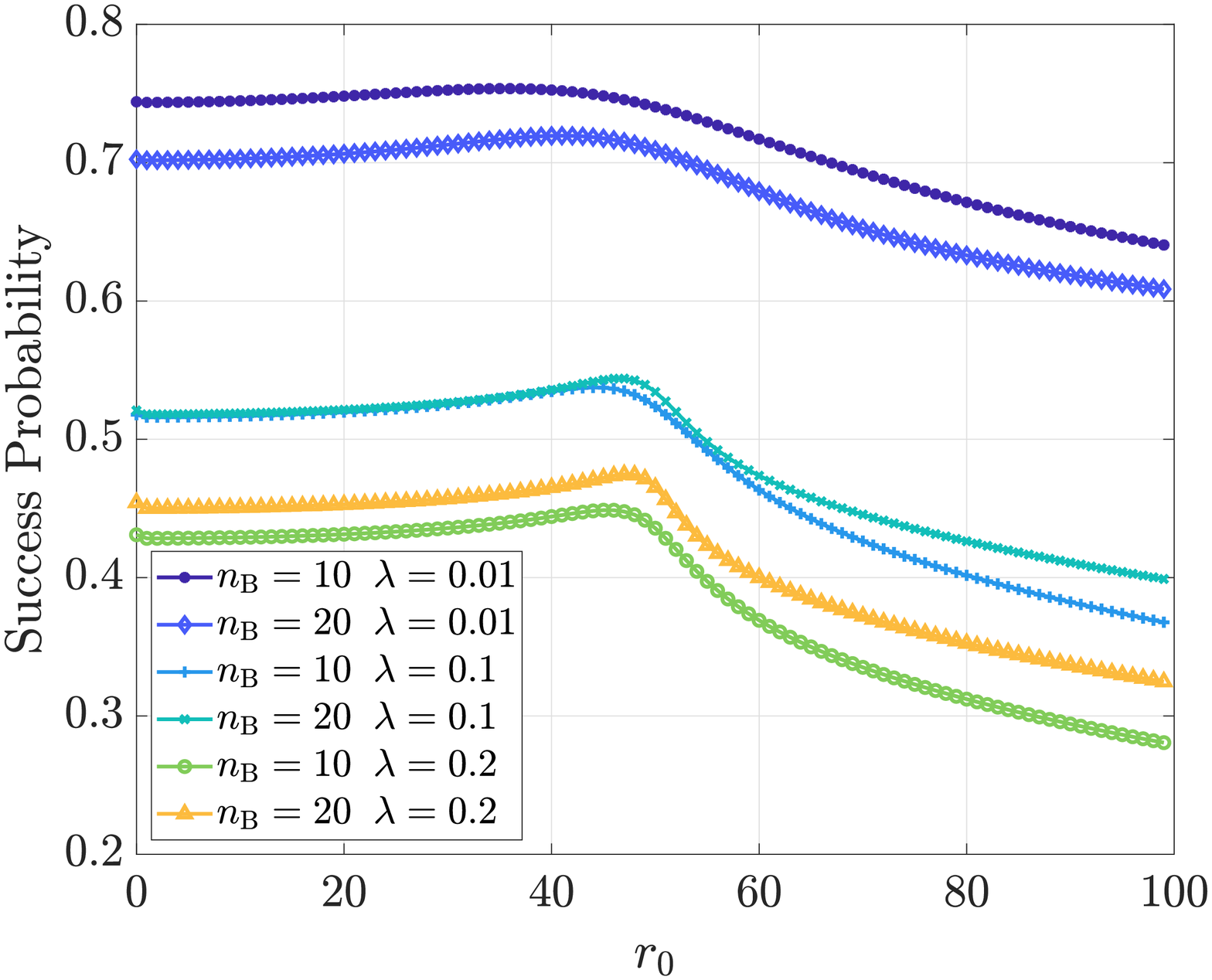}
\label{fig:cov_vs_xt}}
\hfil
\subfloat[]
{\includegraphics[width=0.42\textwidth]{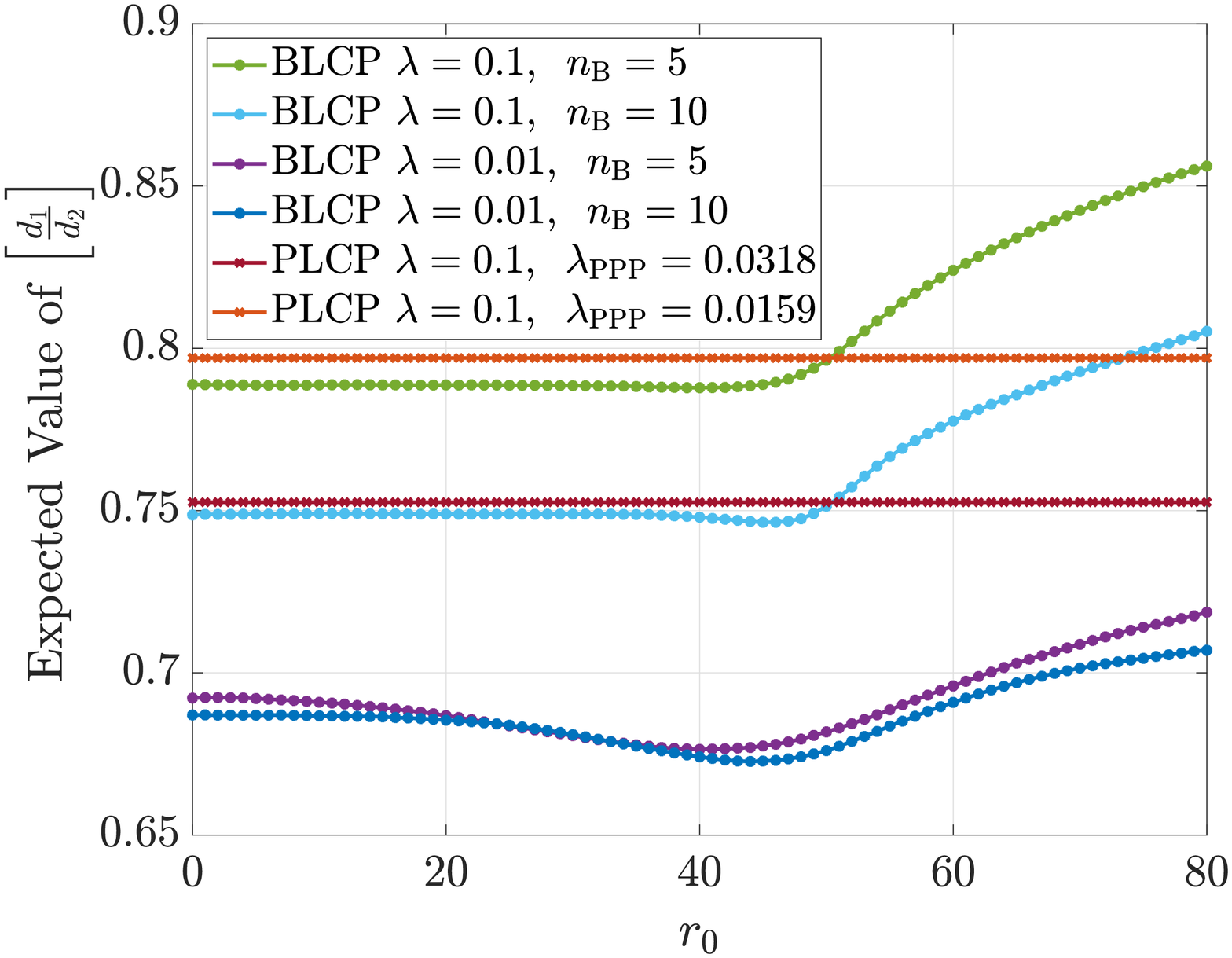}
\label{fig:e_d1_d2}}
\caption{(a) Success probability with respect to $r_0$. Here $R = 50$. (b) $\mathbb{E}\left[\frac{d_1}{d_2}\right]$ with respect to $r_0$.}
\label{fig:10_11} 
\end{figure}

In Fig.~\ref{fig:cov_vs_xt}, we plot the success probability with respect to $r_0$ for different values of $\lambda$ and $n_{\rm B}$. We observe that the success probability first increases slightly with $r_0$, reaches a maximum and starts decreasing. To delve deeper into this phenomenon, we plot the expected value of the ratio of the distance from the nearest BLCP point to the distance of the second nearest BLCP point with respect to $r_0$. In systems where the majority of the interference power is contributed by the nearest interferer, this parameter acts as a good indicator of the success probability. From Fig.~\ref{fig:e_d1_d2}, we note that the value of $\mathbb{E}\left[\frac{d_1}{d_2}\right]$ decreases with $r_0$ at first, reaches its minimum value near $r_0 = R$ and increases beyond that. This indicates that as we move away from the city center, although both the serving \ac{AP} and the interferers become statistically distant from the test point, the relative increase in $d_1$ is higher as compared to the relative increase in $d_2$ with $r_0$. Such an insight for urban networks, in case the streets follow a BLCP cannot be obtained with PLCP models (also shown in Fig.~\ref{fig:e_d1_d2} for different densities).

\subsection{Optimal Network Parameters}
In Fig.~\ref{fig:cov_vs_nb} and Fig.~\ref{fig:cov_vs_lambda} we plot the success probability with respect to $n_{\rm B}$ and $\lambda$ respectively for different values of $r_0$. We observe that for a given location of the test device, increasing the number of lines may increase or decrease the success probability depending on the location of the test device and the density of deployment. In particular, the results suggest that a higher number of \acp{AP} are needed to be deployed in a scenario with dense streets. Based on the location of the test device, there may exist an optimal deployment density of the \ac{AP}s. For example, for a test device located at $r_0 = 50$, a deployment density of $\lambda = 50$ maximizes the coverage. Fig.~\ref{fig:cov_vs_lambda} shows that as $\lambda$ increases the coverage decreases, due to the increased interference. 
After a certain $\lambda$ value, more streets provide better coverage than fewer streets.
Similarly, in Fig.~\ref{fig:cov_vs_nb}(b), we observe that $n_{\rm B}$ maximizes coverage for a given $r_0$ and $\lambda$. Also, for high values of $\lambda$, coverage increases with $n_{\rm B}$ and then decreases.





\begin{figure}[t]
    \centering
    \includegraphics[width = 0.85\textwidth, height = 7cm]{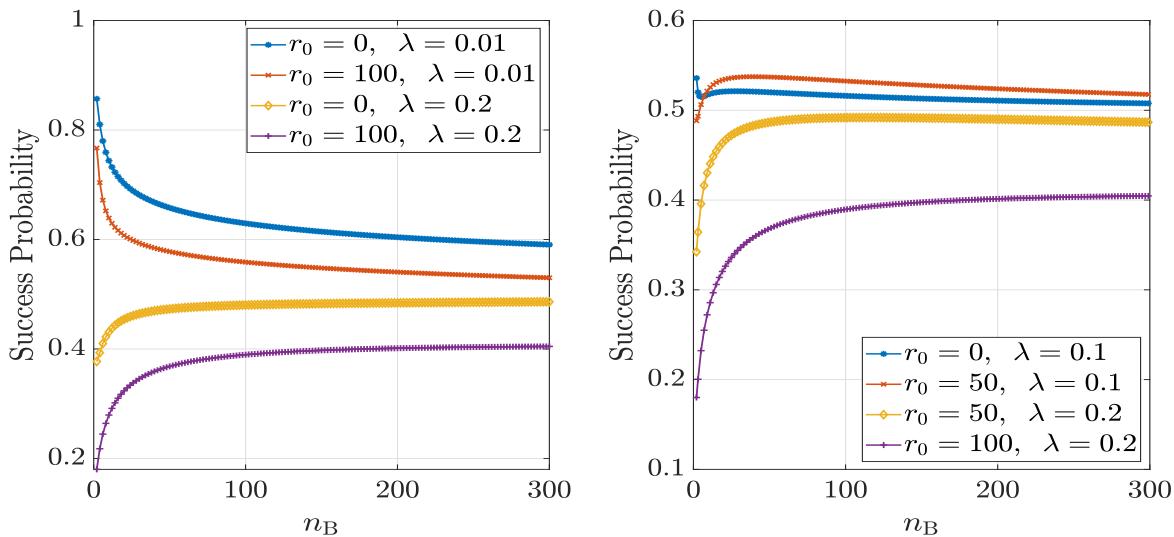}
    \caption{Success probability with respect to $n_{\rm B}$.}
    \label{fig:cov_vs_nb}
\end{figure}
\begin{figure}
    \centering
    \includegraphics[width = 0.45\textwidth]{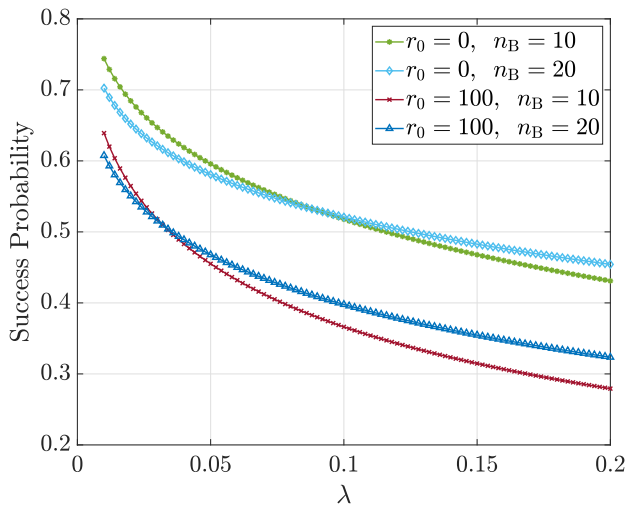}
    \caption{Success probability with respect to the density of \ac{AP}s.}
    \label{fig:cov_vs_lambda}
\end{figure}

\begin{figure}[t]
\centering
\subfloat[]
{\includegraphics[width=0.4\textwidth]{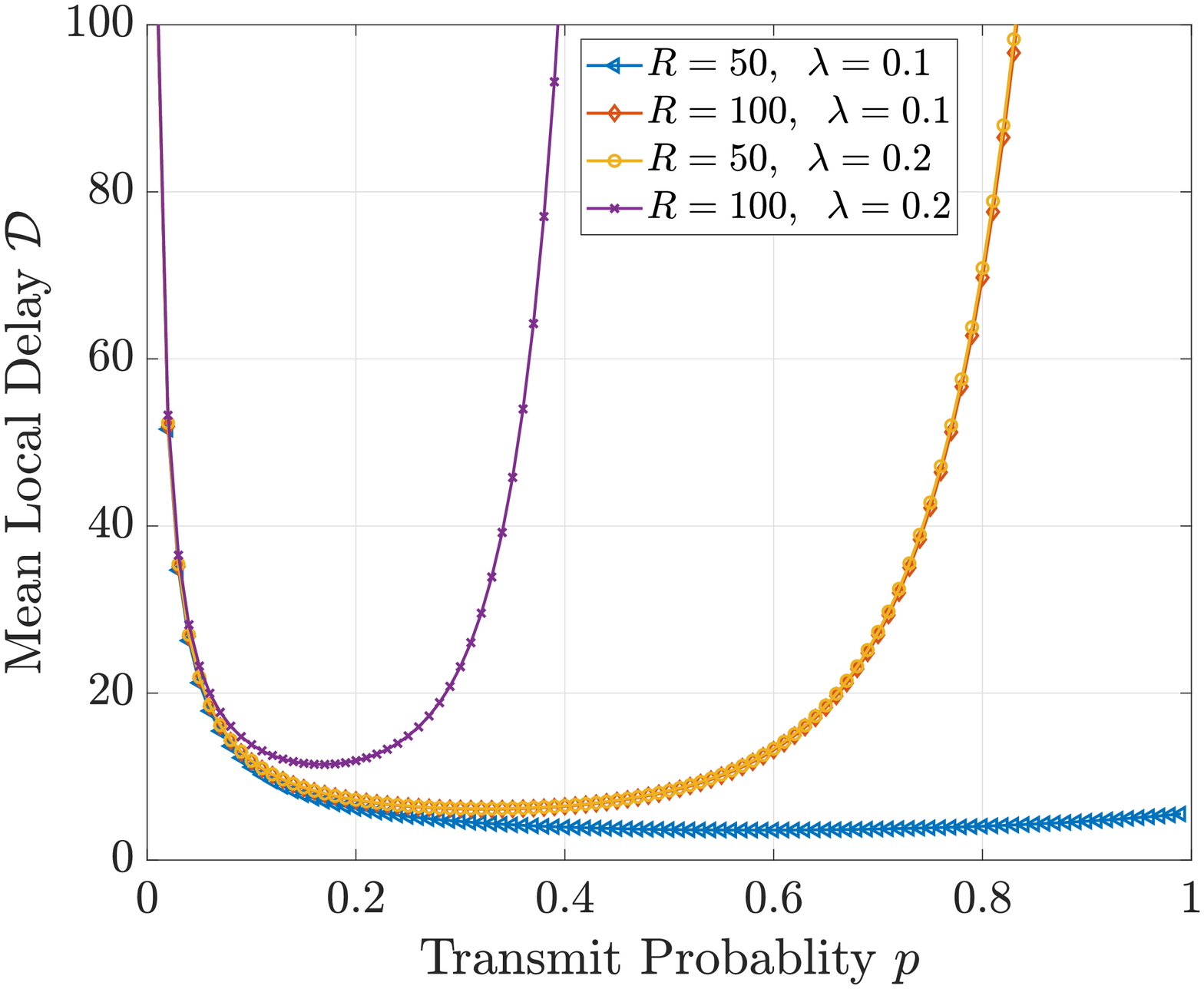}
\label{fig:delay}}
\hfil
\subfloat[]
{\includegraphics[width=0.4\textwidth]{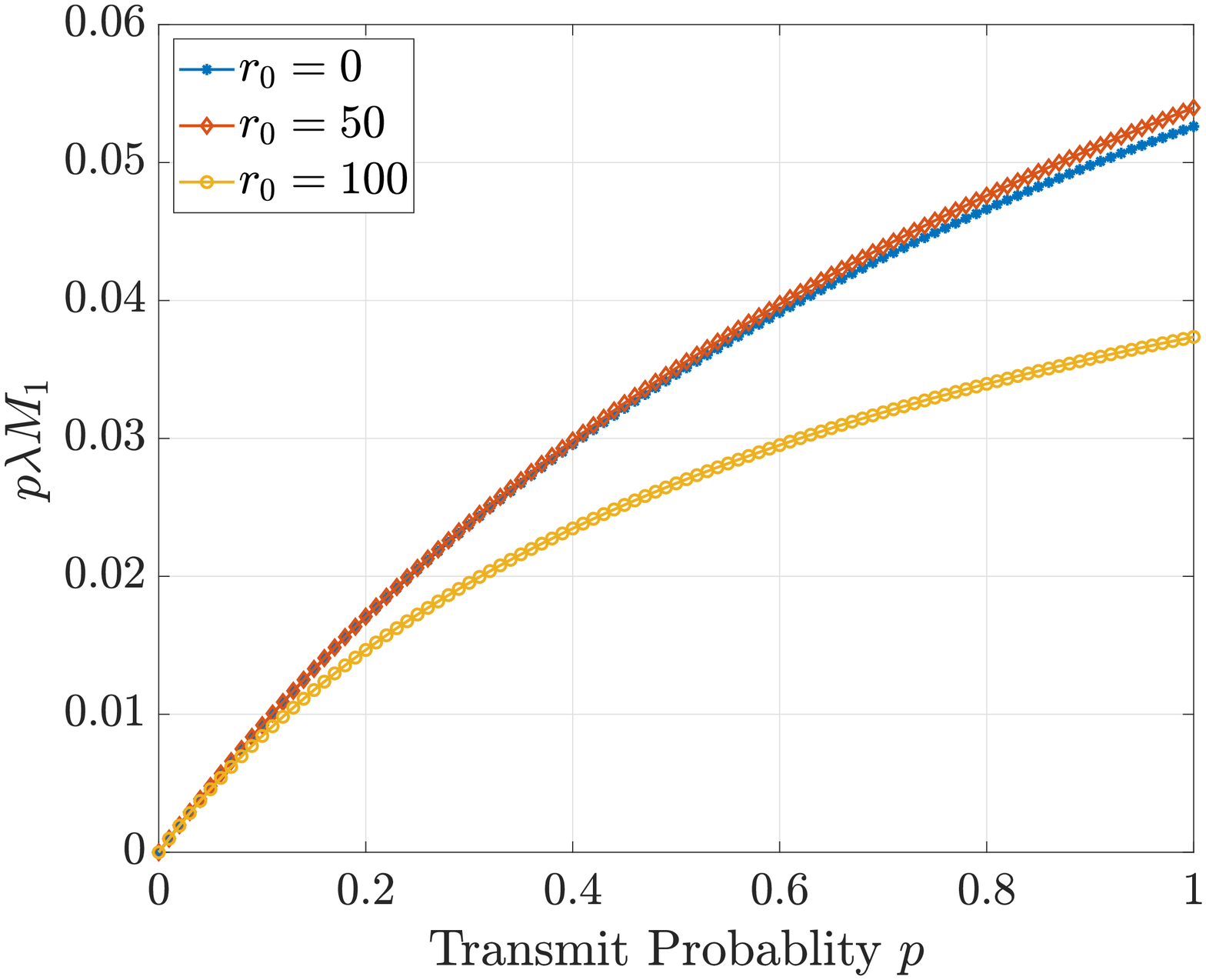}
\label{fig:std}}
\caption{(a) Mean local delay with respect to the transmit probability for different values of $R$ and $\lambda$. Here, $r_0 = 0$. (b) Successful transmission density.}
\label{fig:13_14} 
\end{figure}

\subsection{Results on Moments of Conditional Success Probability and Meta Distribution}
In Fig.~\ref{fig:delay} we plot the mean local delay for a test node present at the origin. As the transmit probability increases, we see that initially, the local delay decreases since the transmitter accesses the channel more frequently. However, as the transmit probability increases, a higher density of interfering \ac{AP}s results in the deterioration of the coverage and accordingly, an increase in the delay. This indicates the existence of optimal transmit probability for minimizing the delay. The optimal transmit probability for minimizing the mean local delay is non-trivial to derive due to the two contending phenomena - increasing $p$ i) increases the frequency with which the service to the test device is attempted thereby reducing the delay and ii) increases the interference which reduces the transmission success and increases the delay. Such nuances of the wireless network are kept for future study.

Fig.~\ref{fig:std} shows the successful transmission density, $p\lambda M_1$, which is the number of successful transmissions per unit area. This acts as an indicator of the network capacity. Here we have assumed $R=50$, $n_{\rm B} = 10$, and $\gamma = -10$ dB. We see that $p\lambda M_1$ increases as $p$ increases specifying that more transmission leads to better transmission density. As we move closer to the city edge i.e., $r_0=50$ successful transmission density increases slightly as compared to $r_0=0$ due to reduced interference. However, at a further distance from the city center, the successful transmission density decreases due to the increasing distance from the serving AP. This is consistent with the results of Fig.~\ref{fig:cov_vs_xt} where we see that as $r_0$ increases coverage first increases then decreases. Such characteristics of the network as a function of $r_0$ cannot be obtained using classical models such as PLP and PLCP.


The optimal transmit probability for minimizing the mean local delay is plotted in Fig. \ref{fig:optdelay} for different locations of the test node, $R$, $\lambda$, and $n_{\rm B}$. We observe that as $r_0$ increases, $p^{\ast}$ also increases first and then decreases. We see that the maximum value of $p^{\ast}$ occurs near the edge of the domain of the BLP. As $r_0$ increases further, the distance to the nearest transmitter as well as the other interfering nodes increases. Consequently, the transmit probability reduces in order to limit the device outage. Interestingly, for $r_0\leq30$, $p^{\ast}$ for $n_{\rm B}=10$ is higher as compared to $n_{\rm B}=20$ since the lines are densely packed around the city center. On the contrary, for $r_0\geq30$, $n_{\rm B}=20$ needs a higher transmission probability than the case with $n_{\rm B}=10$.

\begin{figure}[t]
\centering
\subfloat[]
{\includegraphics[width=0.4\textwidth]{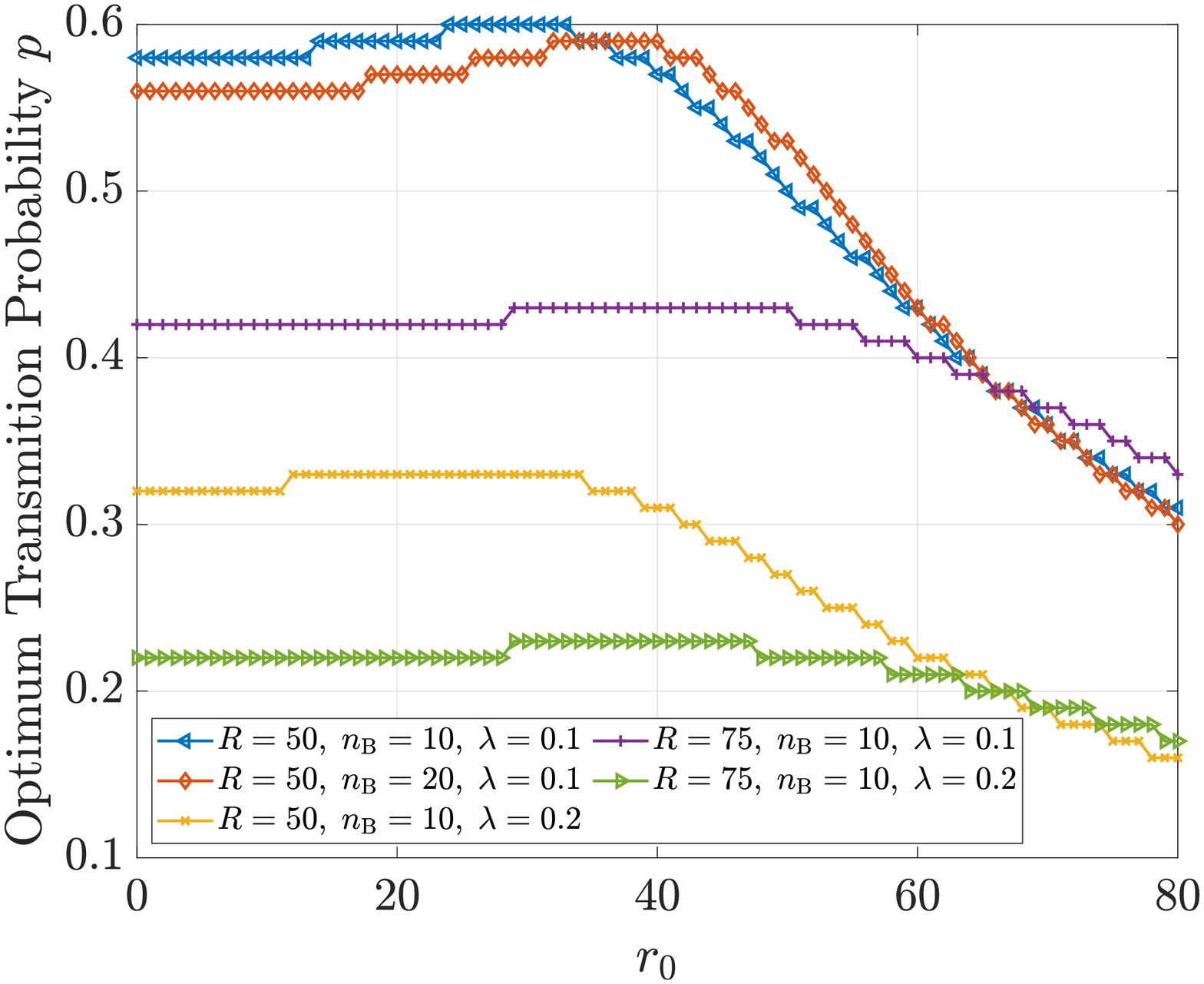}
\label{fig:optdelay}}
\hfil
\subfloat[]
{\includegraphics[width=0.4\textwidth]{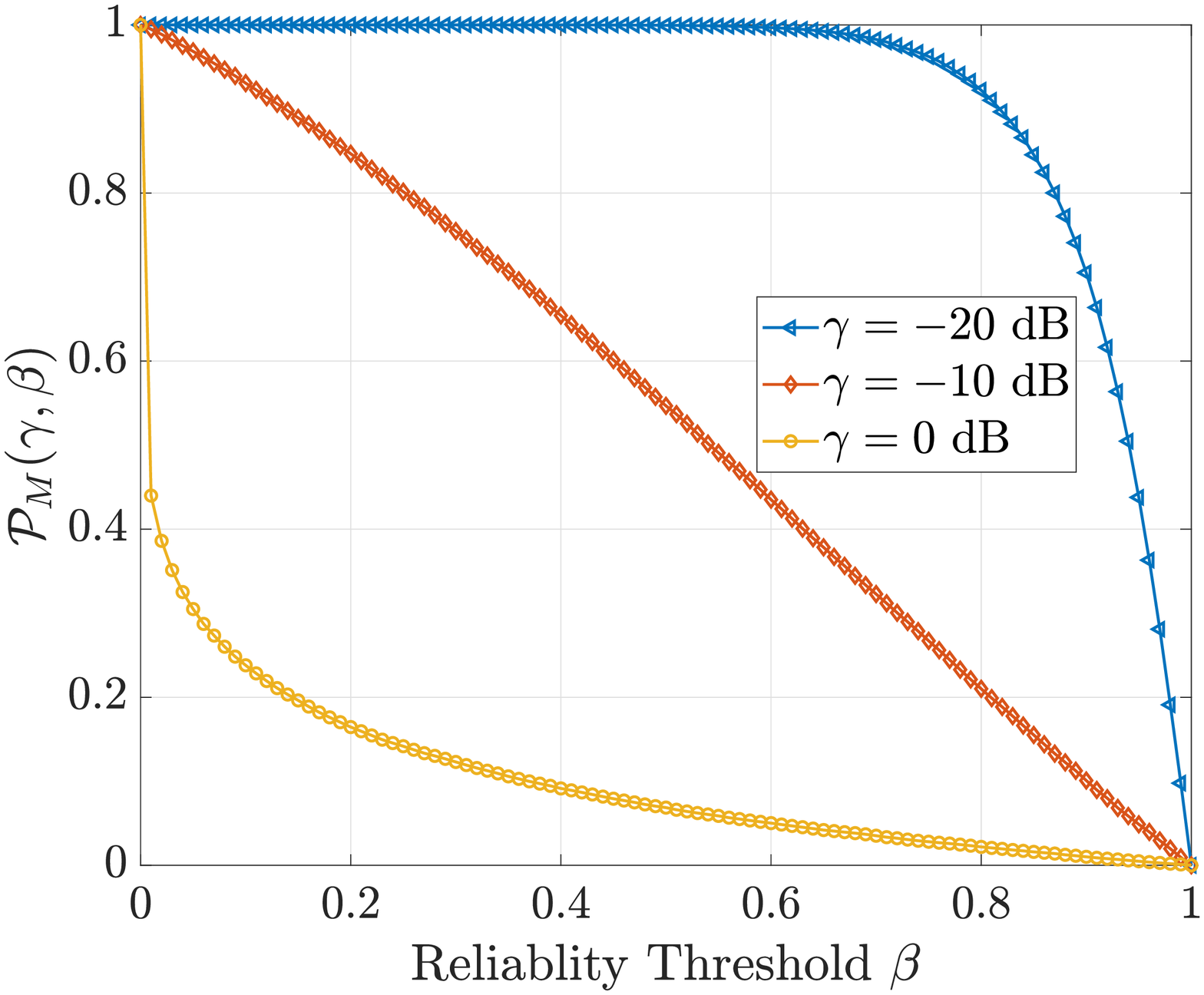}
\label{fig:meta_vs_t}}
\caption{(a) Optimal transmit probability for minimizing the mean local delay. (b) SINR meta distribution.}
\label{fig:15_16} 
\end{figure}

The meta distribution of the success probability is plotted in Fig.~\ref{fig:meta_vs_t} with respect to the reliability threshold $\beta$, for $R=50$ and $r_0=0$. We observe that for an SINR threshold of $\gamma=-20$ dB, the majority of the users are under coverage with reliability of $70\%$ (or probability 0.7). On the contrary, for a service characterized by $\gamma=0$ dB, with a 90\% guarantee, it can be claimed that none of the users will be under coverage. For $\gamma=-10$ dB we see that there is a near-linear relationship between the reliability and the number of users under coverage.


\section{Conclusions and Future Work}
\label{sec:con}
In this paper we have characterized the binomial line Cox process (BLCP) which takes into account the non-homogeneity of lines in a Euclidean plane. Although there are several line processes studied in literature, none of the existing models take into account the non-homogeneity of the lines. This is a drawback of the existing models since practical problems e.g., wireless network planning or transport infrastructure planning need to deal with non-homogeneous streets in a city. We have derived the distribution to the distance of the nearest intersection, the probability generating functional of the BLCP and used it to analyze the transmission success probability in a wireless network. Then, we provide extensive numerical results to derive system design insights for such network deployments. Furthermore, we characterize the meta-distribution of the SINR in order to gain a fine-grained view of the network.
We envisage that the statistical model developed in this paper will be employed in the study of practical problems involving urban street planning. The shortest path-length, multi-tier nodes, and non-homogeneous Poisson line Cox processes are indeed interesting directions of research, which will be taken up as a future work.


\bibliographystyle{ieeetr}
\bibliography{references}

\appendices

\section{Proof of Theorem~\ref{theo:intersection}}
\label{app:intersection}
Let $S = \mathcal{B}\left((0,0), t\right)$ and assume that the BLP is generated within a circle of radius of $R$ with $n_{\rm B}$ lines. In order to determine the average number of intersections over $S$, let us consider a BLP line $L_0$ generated at the point $(0,r_0)$ where $0\leq r_0 \leq \min\{t, R\}$. First, let us determine the domain band $\mathcal{D}_{\times}$ corresponding to the intersection on $L_0$ i.e., all such $(\theta_i,r_i)$ for which lines $L_i \in \mathcal{L}$ will intersect line $L_0$ within $S$. For a given $\theta$, the range of $r$ where if a line is generated intersects $(0,r_0)$ within $S$ is,
\begin{align*}
    \max\left\{-R,\left(r_0\cos\theta - \sqrt{t^2 - r_0^2}\sin\theta\right)\right\} \leq r_i \leq \min\left\{R,\left(r_0 \cos\theta + \sqrt{t^2 - r_0^2} \sin\theta\right)\right\}.
\end{align*}
We see that for values of $t>R$ domain band gets clipped to $R$ and $-R$ in the upper and lower limits, thus the area of the domain band for these two cases i.e. $t \leq R$ and $t > R$ can be derived as follows.
\par \textbf{Case 1:} $t \leq R$. Here $S \subset \mathcal{B}((0,0), R)$. Consequently, there is no clipping in the values of $r$ and the area of the domain band averaged for uniformly distributed $r_0$ between 0 and $t$ can be written as
\begin{align*}
    A_{D_{\times ,1}}(t) &= \frac{1}{t}\int_0^t\int_0^\pi \left(r_0\cos\theta + \sqrt{t^2 - r_0^2}\sin\theta\right) \mathrm{d}\theta\; \mathrm{d}r_0 - \frac{1}{t}\int_0^t\int_0^\pi \left(r_0 \cos\theta - \sqrt{t^2 - r_0^2} \sin\theta\right) \mathrm{d}\theta\; \mathrm{d}r_0= \pi t.
\end{align*}
\par \textbf{Case 2:} $t > R$.  Here $\mathcal{B}((0,0), R) \subset S$. Accordingly, the values of $r$ are limited to $R$ and $-R$. The values of $\theta$ for which $r$ are clipped are
\begin{align}
    \theta_{11}(r_0, t) &= \arctan{\left(\sqrt{\frac{t^2}{r_0^2} - 1}\right)} - \arccos{\left(\frac{R}{t}\right)}, 
    &\theta_{21}(r_0, t) =  \pi - \arccos{\left(\frac{R}{t}\right)} - \arctan{\left(\sqrt{\frac{t^2}{r_0^2} - 1}\right)},\nonumber \\
    \theta_{12}(r_0, t) &= \arctan{\left(\sqrt{\frac{t^2}{r_0^2} - 1}\right)} + \arccos{\left(\frac{R}{t}\right)},
    &\theta_{22}(r_0, t) = \pi + \arccos{\left(\frac{R}{t}\right)} - \arctan{\left(\sqrt{\frac{t^2}{r_0^2} - 1}\right)}. \nonumber
\end{align}
We see that $r$ is clipped to $R$ for values of $\theta$ from $\theta_{11}$ to $\theta_{12}$, similarly $r$ would be clipped to $-R$ for values of $\theta$ from $\theta_{21}$ to $\theta_{22}$. Thus the area of the domain band averaged for uniformly distributed $r_0$ between 0 and $R$ is
\begin{align*}
    A_{D_{\times ,2}}(t) &= \frac{1}{R}\int_0^R \left[2 \int_0^{\theta_{11}} \left(r_0\cos\theta + \sqrt{t^2 - r_0^2}\sin\theta\right)\; \mathrm{d}\theta\; + R(\theta_{21} - \theta_{11})\right]\mathrm{d}r_0 - \\ & \hspace*{3cm} \frac{1}{R}\int_0^R \left[2 \int_0^{\theta_{21}} \left(r_0\cos\theta - \sqrt{t^2 - r_0^2}\sin\theta\right)\; \mathrm{d}\theta\; + R(\theta_{22} - \theta_{21})\right]\mathrm{d}r_0 \\
    & = \frac{2}{R}\left(t^2 \arcsin{\left(\frac{R}{t}\right)} + R\left(2 R \arccos{\left(\frac{R}{t}\right)} - \sqrt{t^2 - R^2} \right)\right).
\end{align*}
Thus, the area of the domain band for a line to intersect the line $L_0$ within $S$ is
\begin{equation}
    A_{D_{\times}}(t) = \begin{cases} 
    \pi t,  &\text{if } t \leq R,    \\[1ex]
    \frac{2}{R}\left(t^2 \arcsin{\left(\frac{R}{t}\right)} + 2 R^2 \arccos{\left(\frac{R}{t}\right)} - \sqrt{t^2 - R^2} \right) & \text{if } t > R.
    \end{cases}
    \label{eq:areaDI}
\end{equation}
Accordingly, the probability that a line of the BLP intersects a single line within $S$ is obtained by $\frac{A_{D_{\times}}(t)}{2\pi \min\{t, R\}}$. It is given as
\begin{align*}
    \mathcal{P}_{\times} (t) = \begin{cases} 
    \frac{1}{2},  &\text{if } t \leq R ,   \\[1ex]
    \frac{1}{\pi R^2} \left(t^2 \arcsin{\left(\frac{R}{t}\right)} + 2 R^2 \arccos{\left(\frac{R}{t}\right)} - \sqrt{t^2 - R^2} \right) & \text{if } t > R.
    \end{cases}
\end{align*}
Now, let us assume that $k$ lines are generated in $S$. Each of these intersects $L_0$ with probability $\mathcal{P}_{\times} (t)$. As a result, the average number of intersections on $L_0$ within $S$ from the $k$ lines is evaluated as

\begin{align*}
    \mathcal{N}^\prime = \sum_{j=0}^k j \binom{k}{j} \left(\mathcal{P}_{\times} (t \,|\, t\leq R)\right)^{j} \left(1 - \mathcal{P}_{\times} (t \,|\, t\leq R\right)^{k-j} = \frac{k}{2}.
\end{align*}

Finally in order to determine the average number of intersections on all the lines within $S$, we take the expectation over the number of lines are generated within $S$. This is evaluated as
\begin{align*}
    \mathcal{N}_1 &= \sum_{k=0}^{n_{\rm B}-1} \underbrace{\binom{n_{\rm B}}{k+1}}_{T_1} \underbrace{\left(\frac{t}{R}\right)^{k+1} \left(1 - \frac{t}{R}\right)^{n_{\rm B}-k-1}}_{T_2} \times \underbrace{\frac{k}{2}}_{T_3} \times \underbrace{(k+1)}_{T_4} \times \underbrace{\frac{1}{2}}_{T_5}= \frac{n_{\rm B} (n_{\rm B} - 1)}{4} \left(\frac{t}{R}\right)^2. 
\end{align*}

where $T_1$ corresponds to choosing $k$ out of $n_{\rm B}$ lines, $T_2$ refers to the probability that exactly $k$ lines are generated in $S$, $T_3$ is the average number of intersections on a single line given that $k$ lines are generated in $S$, $T_4$ is due to the fact that including $L_0$ there are $k+1$ lines in $S$, and finally, $T_5$ is to avoid the double counting of the intersections. Similarly, in case $t > R$ there are $n_{\rm B}$ lines generated, thus the average number of intersections on all lines is
\begin{align*}
    \mathcal{N}_2 = \sum_{k_1=0}^{n_{\rm B} - 1} k_1\mathcal{P}_{\times} (t \,|\, t > R) = \frac{n_{\rm B} (n_{\rm B} - 1)}{2 \pi R^2} \left(t^2 \arcsin{\left(\frac{R}{t}\right)} + R\left(2 R \arccos{\left(\frac{R}{t}\right)} - \sqrt{t^2 - R^2} \right)\right).
\end{align*}
Combining the above, the average number of intersections within a disk of radius $t$ centered at the origin is
\begin{equation}
    \mathcal{N} = \begin{cases} 
    \frac{n_{\rm B} (n_{\rm B} - 1)}{4} \left(\frac{t}{R}\right)^2,  &\text{if } t \leq R   , \\[1ex]
    \frac{n_{\rm B} (n_{\rm B} - 1)}{2 \pi R^2} \left(t^2 \arcsin{\left(\frac{R}{t}\right)} + R\left(2 R \arccos{\left(\frac{R}{t}\right)} - \sqrt{t^2 - R^2} \right)\right) & \text{if } t > R.
    \end{cases}
    \label{eq:intersectionMeasure}
\end{equation}
Next, similar to subsection 2.3, we consider concentric circles centered at the origin having radii $l = \{w, 2w, \dots \}$. As a result, the annuli formed by these concentric circles have the same width, $w$. In the $i-${th} annulus of width $w$, we use (\ref{eq:intersectionMeasure}) to determine the ratio of the average number of intersections to the area to determine intersection radial density as,
\begin{align}
    &\rho_{{\times}, i} (w) = 
    \begin{cases}
        \frac{1}{\pi w^2 (2i+1)} \left(\frac{n_{\rm B} (n_{\rm B} - 1)}{4} \times \left(\frac{(w(i+1))^2}{R^2} - \frac{(wi)^2}{R^2}\right)\right) ; \quad  \hspace{3cm}  \text{for } (i+1)w \leq R, \\
        \frac{n_{\rm B} (n_{\rm B} - 1)}{2\pi^2 R^2 w^2 (2i+1)} \Bigg((i+1)^2w^2 \arcsin{\left(\frac{R}{(i+1)w}\right)} + 2R^2 \arccos{\left(\frac{R}{(i+1)w}\right)} \\
         \hspace{1cm} - R\sqrt{(i+1)^2w^2 - R^2} - i^2w^2 \arcsin{\left(\frac{R}{iw}\right)} - 2R^2 \arccos{\left(\frac{R}{iw}\right)}\\ 
         \hspace{5.4cm} + R\sqrt{i^2w^2 - R^2} \Bigg); \hspace{2cm} \text{for } (i+1)w>R.
    \end{cases}
    \label{eq:intersectionAnnulus}
\end{align}
The final result of the intersection radial density of a BLP can be obtained by substituting $iw = r$ and taking the limit $w \to 0$ in \eqref{eq:intersectionAnnulus}.

\end{document}